\documentclass[preprint,12pt,authoryear]{elsarticle}
\usepackage[T1]{fontenc}
\usepackage[utf8]{inputenc}
\usepackage{amsmath,amssymb,amsthm,mathtools}
\usepackage{booktabs,array,multirow}
\usepackage{graphicx}
\usepackage{algorithm}
\usepackage{algpseudocode}
\algrenewcommand\algorithmicrequire{\textbf{Input:}}
\usepackage{enumitem}
\setlist{nosep,leftmargin=*}
\usepackage{xcolor}
\usepackage[colorlinks=true,linkcolor=black,citecolor=black,urlcolor=blue!50!black]{hyperref}
\usepackage{microtype}
\usepackage{setspace}
\usepackage{adjustbox}
\usepackage{etoolbox}
\AtBeginEnvironment{table}{\linespread{1}\selectfont}
\AtBeginEnvironment{algorithm}{\linespread{1.1}\selectfont\raggedright}

\journal{Computational Statistics and Data Analysis}

\newtheorem{theorem}{Theorem}
\newtheorem{proposition}{Proposition}
\newtheorem{corollary}{Corollary}
\newtheorem{lemma}{Lemma}[section]
\newtheorem{assumption}{Assumption}
\newtheorem{sproposition}[lemma]{Proposition}

\newcommand{\R}{\mathbb{R}}
\newcommand{\E}{\mathbb{E}}
\renewcommand{\Pr}{\mathrm{P}}
\newcommand{\KL}{\mathrm{KL}}
\newcommand{\hU}{\widehat U}
\newcommand{\hPi}{\widehat\Pi_n}
\newcommand{\hM}{\widehat M_n}
\newcommand{\Rset}{\mathcal R_{d,\kappa}}

\hypersetup{pdftitle={Bayesian heterogeneous copula mixtures with nonparametric margins: consistency, identifiability and tail asymmetry in physical fitness data},
pdfauthor={Yujian Liu}}

\begin{document}

\begin{frontmatter}

\title{Bayesian heterogeneous copula mixtures with nonparametric margins: consistency, identifiability and tail asymmetry in physical fitness data}

\author[a,b]{Yujian Liu\corref{cor1}}
\ead{liuyujian02@sus.edu.cn}
\cortext[cor1]{Corresponding author. ORCID: \href{https://orcid.org/0000-0002-9512-1272}{0000-0002-9512-1272}.}

\affiliation[a]{organization={School of Intelligent Sports Engineering, Shanghai University of Sport},
            addressline={399 Changhai Road, Yangpu District},
            city={Shanghai},
            postcode={200438},
            country={China}}
\affiliation[b]{organization={School of Economics and Management, Shanghai University of Sport},
            addressline={399 Changhai Road, Yangpu District},
            city={Shanghai},
            postcode={200438},
            country={China}}

\begin{abstract}
Finite mixtures of Clayton, Gumbel, Frank and Gaussian copulas can describe dependence that differs between the upper and lower tails. We study a two-stage Bayesian analysis in which ranks or kernel estimates replace the margins and the copula likelihood is evaluated at the resulting pseudo-observations. This pseudo-posterior is strongly consistent for the copula density whenever the log density admits a logarithmic boundary envelope. Both stages may use the same data, margins may be standardized within observed strata, and neither smoothness nor identifiability is required. The envelope holds for finite mixtures of Gaussian, Student, Clayton, Gumbel and Frank copulas in any fixed dimension; tail-dependence coefficients and conditional tail probabilities are therefore consistently estimated. We further prove that Clayton, Gumbel, Frank and Gaussian copulas are jointly finitely linearly independent, which makes mixture weights and components identifiable and consistently estimated. In simulations the pseudo-posterior matches multi-start maximum pseudo-likelihood in large samples. It is more stable in small samples and near independence (every component close to the independence copula), where tail coefficients are learned long before the weights and a marginal Metropolis sampler is up to twice as efficient as data augmentation. Two physical fitness datasets show mirror-image asymmetries. Among 8772 university students, sprint and jump performance are coupled mainly at the top. Among 5336 adults in a national health survey, low grip strength and low daily activity cluster together, increasingly with age, whereas high values do not. A Gaussian copula misses both patterns in held-out data.
\end{abstract}


\end{frontmatter}

\section{Introduction}\label{sec:intro}

Fitness and health measurements are usually interpreted one at a time. Many practical questions, however, concern how they co-occur at the extremes. Do the students who sprint fastest also jump farthest, and are the slowest also the weakest jumpers? Are adults with a weak grip also the least active, and are the strongest also the most active? Such questions concern the tails of a joint distribution, which a correlation coefficient does not describe: a rank correlation near zero is compatible with a strong co-occurrence of low values that is offset elsewhere in the distribution. Copulas separate dependence from the margins. If $X\in\R^d$ has continuous marginal distribution functions $F_1,\ldots,F_d$, then $\Pr(X\le x)=C\{F_1(x_1),\ldots,F_d(x_d)\}$ for a unique copula $C$ \citep{Sklar1959,Nelsen2006}. A single parametric copula, however, fixes the tails in advance. Gaussian and Frank copulas have no tail dependence and Student $t$ copulas the same dependence in both tails, while Clayton and Gumbel copulas have dependence only in the lower and only in the upper tail, respectively \citep{Nelsen2006,Joe2014}. Choosing one of them presupposes the answer. A finite mixture $c=\sum_kw_kc_{\theta_k}$ that combines several families leaves the answer to the data. In a mixture of Clayton, Gumbel, Frank and Gaussian copulas, the four-family mixture for short, the weights indicate whether dependence concentrates in the lower tail, in the upper tail or in neither. Because all components share the same margins, they describe patterns of co-movement, not levels of performance \citep{Kosmidis2016,QuLu2021}.

The margins are often estimated nonparametrically. Ranks and kernel estimates adapt to skewed, heavy-tailed or multimodal margins without restricting their shape, and they protect the dependence analysis from misspecified margins, which can bias copula estimates considerably \citep{Kim2007}. The copula is then fitted to pseudo-observations $\hU_i$ obtained from ranks or smoothed marginal distribution functions. For a single parametric copula such two-stage estimators are well understood \citep{Genest1995,Joe2005,Tsukahara2005,ChenFan2006,FermanianPoignard2024}; mixtures are fitted by penalized or expectation--maximization type likelihood methods \citep{CaiWang2014,Kauermann2014,LiuLong2022}. A Bayesian analysis can take the same route and use the copula density at the transformed observations as a likelihood. This has been done for pair-copula constructions with empirical margins \citep{MinCzado2010}, for time-varying copulas with parametric margins \citep{AlmeidaCzado2012} and for copula mixtures \citep{Liu2023}, and Bayesian nonparametric copula mixtures have also been proposed \citep{Pan2025,LiuYu2026}. The resulting pseudo-posterior is easy to sample. It quantifies uncertainty about weights, component parameters and derived quantities such as tail-dependence coefficients. For copula mixtures, however, its large-sample behavior is unknown.

Establishing that behavior raises two difficulties. The first concerns the margins. The pseudo-observations are dependent and only approximately uniform, and copula log densities are unbounded at the boundary of the unit cube, exactly where marginal errors are largest relative to the values being estimated. Results for generalized posteriors \citep{Bissiri2016,Miller2021} take a limiting criterion as given; they do not show that it survives the replacement of $U_i$ by $\hU_i$. The second concerns the components, about which a consistent estimate of a mixture density says nothing unless the mixture is identifiable. Asymptotic theory for copula mixtures assumes identifiability \citep{CaiWang2014,LiuLong2022}. For mixtures that combine Archimedean and elliptical families, this assumption has, to our knowledge, never been verified.

Theorem~\ref{thm:consistency} resolves the first difficulty. The two-stage pseudo-posterior is strongly consistent for the copula density in $L^1$ whenever the log densities are bounded, uniformly over a compact parameter set, by the logarithmic envelope $A+B\sum_j[-\log\{u_j(1-u_j)\}]$. This envelope holds for finite mixtures of Gaussian, Student, Clayton, Gumbel and Frank copulas (Corollary~\ref{cor:families}). Every quantity that depends continuously on the parameter and only through the density, tail-dependence coefficients and conditional tail probabilities among them, is then consistently estimated without any identification condition (Proposition~\ref{prop:functionals}). For the second difficulty, Theorem~\ref{thm:identification} shows that exchangeable Clayton, Gumbel and Frank copulas and Gaussian copulas with arbitrary nonsingular correlation matrices are jointly finitely linearly independent in every dimension. The pseudo-posterior therefore also recovers the weights and component parameters of mixtures over this dictionary, and any family absent from the true copula loses its weight (Corollaries~\ref{cor:mixing} and~\ref{cor:fourfamily}).

We sample the four-family pseudo-posterior with a marginal Metropolis sampler, which works on the mixture parameters directly, without latent allocations. It is compared with the standard data-augmentation sampler and with multi-start maximum pseudo-likelihood, all three using the same pseudo-observations, mixture order and parameter bounds. The simulations show how far identifiability carries in finite samples. With well-separated components the weights are recovered. Near independence, where every component copula is close to the independence copula and the components are hard to tell apart, the density and the tail coefficients are estimated accurately long before the weights. In the applications, posterior tail coefficients are complemented by conditional tail probabilities at finite levels, which are easier to interpret and can be checked against held-out data. The two applications give mirror-image results. Among 8772 university students, within-sex standardization reverses an apparent trade-off between power and flexibility, and sprint and jump performance are coupled mainly at the top. Among 5336 adults from the US National Health and Nutrition Examination Survey (NHANES), grip strength and accelerometer-measured daily activity are almost uncorrelated, yet low values of the two cluster strongly, and more so with age.

Section~\ref{sec:model} defines the pseudo-posterior and Section~\ref{sec:theory} states the asymptotic results. Computation, simulations and applications follow in Sections~\ref{sec:computation}, \ref{sec:simulation} and~\ref{sec:applications}, and Section~\ref{sec:discussion} discusses limitations. The Supplementary Material contains the proofs, further simulation results, data-processing details, diagnostics and an additional application to wearable-device data.

\section{Two-stage pseudo-posterior}\label{sec:model}

Let $X_1,\ldots,X_n$ be independent copies of $X\in\R^d$, with fixed $d\ge2$ and continuous marginal distribution functions $F_1,\ldots,F_d$. The vector $U=(F_1(X_1),\ldots,F_d(X_d))$ has uniform margins; we assume that its distribution, the copula of $X$, has a density $c_0$ on $(0,1)^d$, and write $U_i$ for the unobserved value of $U$ belonging to $X_i$.

The working model is a family $\{q_\eta:\eta\in H\}$ of strictly positive copula densities indexed by a compact metric space $H$. The main example is the finite mixture
\begin{equation}\label{eq:mixture}
q_\eta(u)=\sum_{k=1}^K w_kc_{\theta_k}(u),\qquad \eta=(w,\theta_1,\ldots,\theta_K)\in\Delta_{K-1}\times\Theta^K,
\end{equation}
where $\Delta_{K-1}$ is the unit simplex and $\{c_\theta:\theta\in\Theta\}$ is a union of parametric copula families. Zero weights and repeated components are allowed, so~\eqref{eq:mixture} contains every mixture with at most $K$ components.

The margins are estimated from the same sample. We consider scaled ranks $\hU_{ij}=R_{ij}/(n+1)$, where $R_{ij}$ is the rank of $X_{ij}$ among $X_{1j},\ldots,X_{nj}$, and in-sample kernel estimates
\begin{equation}\label{eq:kernel}
\hU_{ij}=\widehat F_j(X_{ij}),\qquad \widehat F_j(x)=\frac1n\sum_{l=1}^n\mathcal K_j\!\left(\frac{x-X_{lj}}{h_{nj}}\right),
\end{equation}
where $\mathcal K_j$ is a distribution function with $\mathcal K_j(0)\in(0,1)$ and the bandwidth $h_{nj}>0$ may depend on the data. Given a prior $\Pi$ on $H$, the two-stage pseudo-posterior is
\begin{equation}\label{eq:posterior}
\hPi(d\eta)=\frac{\prod_{i=1}^nq_\eta(\hU_i)\,\Pi(d\eta)}{\int_H\prod_{i=1}^nq_\zeta(\hU_i)\,\Pi(d\zeta)},
\end{equation}
the posterior obtained by treating $\hU_1,\ldots,\hU_n$ as independent draws from $q_\eta$. Let $\hM(\eta)=n^{-1}\sum_i\log q_\eta(\hU_i)$, $M(\eta)=\E\log q_\eta(U)$ and $b(t)=-\log\{t(1-t)\}$ for $t\in(0,1)$.

\section{Asymptotic theory}\label{sec:theory}

This section states the results and outlines the main ideas of their proofs; complete proofs are given in Section~S1 of the Supplementary Material.

\subsection{Consistency}

\begin{assumption}[Boundary envelope]\label{ass:envelope}
The map $(\eta,u)\mapsto\log q_\eta(u)$ is continuous on $H\times(0,1)^d$, and for some constants $A,B\ge0$,
\begin{equation}\label{eq:envelope}
\sup_{\eta\in H}|\log q_\eta(u)|\le E(u):=A+B\sum_{j=1}^db(u_j),\qquad u\in(0,1)^d.
\end{equation}
\end{assumption}

\begin{assumption}[Correct specification and prior support]\label{ass:support}
There is $\eta_0\in H$ with $q_{\eta_0}=c_0$, and $\Pi(N)>0$ for every open neighborhood $N$ of $\eta_0$.
\end{assumption}

\begin{theorem}\label{thm:consistency}
Let Assumptions~\ref{ass:envelope} and~\ref{ass:support} hold, and let $\hU_i$ be the scaled ranks or the kernel estimates~\eqref{eq:kernel} with $h_{nj}\to0$ almost surely for each $j$. Then, almost surely:
\begin{enumerate}[label=(\alph*)]
\item $\sup_{\eta\in H}|\hM(\eta)-M(\eta)|\to0$;
\item $\hPi\{\eta:\|q_\eta-c_0\|_1>\varepsilon\}\to0$ for every $\varepsilon>0$;
\item the posterior predictive density $\overline q_n=\int q_\eta\,\hPi(d\eta)$ satisfies $\|\overline q_n-c_0\|_1\to0$.
\end{enumerate}
\end{theorem}

The envelope~\eqref{eq:envelope} is the only requirement on the copula family beyond continuity. It allows logarithmic divergence at the boundary and involves no derivatives. The key step bounds the empirical second moment of $E(\hU_i)$ (Lemma~S1.1 in the Supplementary Material), which then controls the observations near the boundary (Lemma~S1.2); for ranks the moment is a Riemann sum, and for kernel estimates monotonicity of $\mathcal K_j$ reduces the bound to the rank case, so the bandwidth need only tend to zero. Identifiability is not assumed: Assumption~\ref{ass:support} may hold at a representation of $c_0$ with zero weights or repeated components.

Fitness and health measurements are usually compared within sex and age group, whose marginal distributions differ sharply. Theorem~\ref{thm:consistency} extends to such standardization. If ranks or kernel estimates are computed within the strata of an observed variable $S$ with finitely many values, its conclusions hold with target $c_0=\sum_s\Pr(S=s)c_{0\mid s}$, the stratum-weighted average of the conditional copula densities (Proposition~S1.3 in the Supplementary Material). Assumption~\ref{ass:support} must then hold for this pooled density. If, for instance, two strata follow Clayton copulas with different parameters, $c_0$ is a two-component Clayton mixture, which a model with one component per family cannot represent exactly.

\subsection{Families and derived quantities}

For $0<\kappa<1$ let $\Rset$ be the set of $d\times d$ correlation matrices whose smallest eigenvalue is at least $\kappa$.

\begin{corollary}\label{cor:families}
Fix $d\ge2$ and $K\ge1$. Let $\Theta$ be a finite union of the following compact component sets: Gaussian copulas with $R\in\Rset$; Student $t$ copulas with $R\in\Rset$ and $\nu\in[\nu_1,\nu_2]\subset(0,\infty)$; exchangeable Clayton, Gumbel and Frank copulas with parameters in $[0,M_{\mathrm C}]$, $[1,M_{\mathrm G}]$ and $[0,M_{\mathrm F}]$, where the endpoints $0$, $1$ and $0$ denote independence and, for $d=2$, the Frank interval may be any compact interval. Then the mixture model~\eqref{eq:mixture} satisfies Assumption~\ref{ass:envelope}. If $\Pi$ is the product of a Dirichlet distribution with positive parameters on $\Delta_{K-1}$ and distributions with full support on $\Theta$, the prior condition in Assumption~\ref{ass:support} holds at every $\eta_0$.
\end{corollary}

Student components may have different degrees of freedom, and no lower bound on the weights is needed, since convex combinations inherit the envelope. Compactness cannot be dropped: no uniform envelope exists if Gaussian correlations may approach $\pm1$ or Clayton parameters may diverge.

\begin{proposition}[Derived quantities]\label{prop:functionals}
Let $g$ be a continuous map from $H$ into a metric space with $g(\eta)=g(\eta')$ whenever $q_\eta=q_{\eta'}$. Under the conditions of Theorem~\ref{thm:consistency}, $\hPi\{\eta:\operatorname{dist}(g(\eta),g(\eta_0))>\varepsilon\}\to0$ almost surely for every $\varepsilon>0$.
\end{proposition}

Many quantities of interest have this form. They are defined through the copula itself and, for the families of Corollary~\ref{cor:families}, depend continuously on $\eta$ (Supplementary Section~S1.6): Kendall's $\tau$ and Spearman's $\rho$ of any pair, conditional probabilities $\Pr(U_d>t\mid U_{-d}=a)$ at interior points $a$, the conditional tail probabilities
\begin{equation}\label{eq:chi}
\begin{aligned}
\chi_L(t)&=\Pr(U_2\le t\mid U_1\le t)=\frac{C(t,t)}t,\\
\chi_U(t)&=\Pr(U_2>1-t\mid U_1>1-t)=\frac{2t-1+C(1-t,1-t)}t
\end{aligned}
\end{equation}
of a bivariate margin at fixed $t\in(0,1)$ \citep{Coles1999}, and their limits as $t\downarrow0$, the lower and upper tail-dependence coefficients $\lambda_L$ and $\lambda_U$. For a bivariate mixture, $\lambda_L=\sum_kw_k\lambda_L(c_{\theta_k})$ and $\lambda_U=\sum_kw_k\lambda_U(c_{\theta_k})$; Clayton components contribute $\lambda_L=2^{-1/\theta}$, Gumbel components contribute $\lambda_U=2-2^{1/\gamma}$, and Frank and Gaussian components contribute nothing \citep{Nelsen2006,Joe2014}. Tail coefficients are not continuous functionals of the density on the set of all copulas, and compactness of the parameter set does not by itself make them so. What Proposition~\ref{prop:functionals} uses is that, for the families of Corollary~\ref{cor:families}, the component coefficients and hence $\lambda_L$ and $\lambda_U$ are continuous in $\eta$.

\subsection{Identifiability of heterogeneous mixtures}

Weights and component parameters are functions of the density only if the mixture is identifiable, that is, if any finitely many distinct members of the component family are linearly independent \citep{Yakowitz1968}. For Gaussian copulas, and for Student copulas with a common known $\nu$, this follows from classical results for elliptical mixtures \citep{Yakowitz1968,Holzmann2006}. Archimedean and elliptical families, however, share the independence copula and have no common location--scale form.

\begin{theorem}\label{thm:identification}
For every $d\ge2$, the family consisting of the independence copula, the exchangeable Clayton copulas with $\theta>0$, the Gumbel copulas with $\gamma>1$, the Frank copulas with $\theta>0$ (any $\theta\ne0$ if $d=2$) and the Gaussian copulas with nonsingular correlation matrices is finitely linearly independent.
\end{theorem}

The proof separates the families by their boundary behavior. As $u\downarrow0$ with $v$ fixed, Frank, Gumbel and Clayton or Gaussian densities have distinct limits or decay rates, and the remaining Clayton and Gaussian terms are ordered along the curve $(\Phi(-x),\Phi(1))$; in higher dimensions, integrating out all but two coordinates removes the Archimedean terms, and the Gaussian terms are then separated in the full dimension.

To state the consequences, identify parameter values that give the same density (in particular the four representations of independence), so that the components form a compact set $\mathcal C=\{c_\theta:\theta\in\Theta\}$ of $L^1$. Let $G_\eta=\sum_kw_k\delta_{c_{\theta_k}}$ be the mixing measure of~\eqref{eq:mixture} and $W$ the Wasserstein distance on probability measures on $\mathcal C$ with the $L^1$ metric \citep{Nguyen2013}.

\begin{corollary}\label{cor:mixing}
Let $\Theta$ be as in Corollary~\ref{cor:families}, consisting either of components from the families in Theorem~\ref{thm:identification}, or of Gaussian components, or of Student components with a common known $\nu$. Under the conditions of Theorem~\ref{thm:consistency}, almost surely:
\begin{enumerate}[label=(\alph*)]
\item $\hPi\{\eta:W(G_\eta,G_{\eta_0})>\varepsilon\}\to0$ for every $\varepsilon>0$;
\item if $G_{\eta_0}$ has exactly $K$ atoms, the weights and components converge up to relabeling: for every $\varepsilon>0$, the pseudo-posterior probability of
\[
\min_\sigma\max_k\bigl\{|w_{\sigma(k)}-w_{0k}|+\|c_{\theta_{\sigma(k)}}-c_{\theta_{0k}}\|_1\bigr\}>\varepsilon
\]
tends to zero, the minimum being over permutations $\sigma$;
\item if a closed set $\mathcal C_F\subset\mathcal C$ contains no atom of $G_{\eta_0}$, then $\hPi\{\eta:G_\eta(\mathcal C_F)>\varepsilon\}\to0$ for every $\varepsilon>0$.
\end{enumerate}
\end{corollary}

By Theorem~\ref{thm:identification}, $G\mapsto q_G$ is a homeomorphism on mixing measures with at most $K$ atoms, so (a) follows from Proposition~\ref{prop:functionals}; (b) and (c) follow from (a) through the Kantorovich--Rubinstein representation of $W$. Under (b), the class probabilities $w_kc_{\theta_k}(u)/q_\eta(u)$ also converge (Proposition~S1.4).

Parts (a) and (c) hold even when $K$ exceeds the true number of components, whereas (b) requires the exact order. If $\mathcal C_F$ collects the components of one family, $G_\eta(\mathcal C_F)$ is that family's total weight, and (c) says that a family absent from $c_0$ loses its weight, whatever the positive Dirichlet parameter. In this soft sense the posterior selects families, although no rate is implied. The condition fails if $c_0$ has an independence component and $\mathcal C_F$ contains independence, a limit shared by all four families; near independence the weight vanishes slowly. Overfitting within one family is not covered: a redundant component can duplicate a true one and share its weight, and whether it is emptied depends on the Dirichlet parameter \citep{Rousseau2011}. The parameter of a component whose weight tends to zero is not identified. Finally, because all components share the margins of $U$, a component describes how variables move together, not how well a person performs on any single test. The applications therefore report weights and tail quantities, not cluster memberships.

The four-family mixture used in the rest of the paper fixes one family per component. For this parameterization the conclusions of Corollary~\ref{cor:mixing} can be stated without relabeling, family by family.

\begin{corollary}[Four-family mixture]\label{cor:fourfamily}
Let $q_\eta=w_1c^{\mathrm C}_{\theta_1}+w_2c^{\mathrm G}_{\theta_2}+w_3c^{\mathrm F}_{\theta_3}+w_4c^{\mathrm N}_{\theta_4}$ be a mixture of exchangeable Clayton, Gumbel and Frank copulas and a Gaussian copula, with $\eta\in H=\Delta_3\times\Theta_1\times\cdots\times\Theta_4$, where $\Theta_1=[a_1,b_1]$ with $a_1>0$, $\Theta_2=[a_2,b_2]$ with $a_2>1$, $\Theta_3=[a_3,b_3]$ with $a_3>0$, and $\Theta_4\subset\Rset$ is compact. Let $\Pi$ be the product of a Dirichlet distribution with positive parameters and distributions with full support on $\Theta_1,\ldots,\Theta_4$, let $c_0=q_{\eta_0}$ with $\eta_0\in H$, and let $\hU_i$ be as in Theorem~\ref{thm:consistency}. Then the conclusions of Theorem~\ref{thm:consistency} and Proposition~\ref{prop:functionals} hold and, almost surely, for $k=1,\ldots,4$ and every $\varepsilon>0$:
\begin{enumerate}[label=(\alph*)]
\item $\hPi\{\eta:|w_k-w_{0k}|>\varepsilon\}\to0$;
\item if $w_{0k}>0$, then $\hPi\{\eta:\|\theta_k-\theta_{0k}\|>\varepsilon\}\to0$.
\end{enumerate}
\end{corollary}

Excluding independence from the three Archimedean parameter sets makes the four families pairwise disjoint compact sets of densities in $L^1$, by Theorem~\ref{thm:identification}. A component can then approximate only its own family, so no relabeling is needed, and (a) also covers a family that is absent from $c_0$. The priors of Section~\ref{sec:priors} satisfy these conditions, and the weights reported in Section~\ref{sec:applications} therefore estimate family-specific quantities. As before, no rate is implied.

\section{Computation}\label{sec:computation}

\subsection{Priors and sampler}\label{sec:priors}\label{sec:sampler}

All fits use a Dirichlet$(\alpha,\ldots,\alpha)$ prior on the weights and independent uniform priors on compact component parameter sets, so Corollary~\ref{cor:families} applies. For the four-family mixture we take $\alpha=1$. The parameter intervals are $[0.01,12]$ for Clayton, $[1.001,7]$ for Gumbel, $[0.01,30]$ for Frank and $[-0.95,0.95]$ for the Gaussian correlation; together they cover Kendall's $\tau$ up to about $0.86$, and only the Gaussian component may carry negative dependence. Priors for the other experiments are listed in Supplementary Section~S2.

The standard sampler for mixtures, data augmentation, alternates between latent allocations, weights and component parameters \citep{Diebolt1994}. When copula components overlap, as they do near independence, the allocations are nearly random and the weights move slowly. We therefore sample the pseudo-posterior~\eqref{eq:posterior} of $\eta=(w,\theta)$ directly, without allocations. This marginal sampler (Algorithm~\ref{alg:sampler}) works in unconstrained coordinates: additive log-ratios of the weights and logits of the parameters' relative positions in their prior intervals. Occasional independence proposals from the prior let a component with small weight, whose parameter the data barely determine, move between distant regions. Proposal scales are adapted during burn-in only, so the retained draws come from a fixed Metropolis--Hastings kernel that leaves the pseudo-posterior invariant \citep{Haario2001,Roberts2009}. An iteration costs about twice as much as one of data augmentation; Section~\ref{sec:exp3} shows when this pays off. Tuning details are given in Supplementary Section~S2.

\begin{algorithm}[!htbp]
\small
\caption{Marginal Metropolis sampler for the pseudo-posterior of a $K$-component copula mixture}\label{alg:sampler}
\begin{algorithmic}[1]
\Require pseudo-observations $\hU_1,\ldots,\hU_n$; families $c^{(1)},\ldots,c^{(K)}$ with parameter intervals $[a_k,b_k]$; iterations $T$, burn-in $T_0<T$; period $m$ of independence proposals ($m=10$)
\State Initialize $(w,\theta)$; set $y_k=\log(w_k/w_K)$ for $k<K$ and $z_k=\operatorname{logit}\{(\theta_k-a_k)/(b_k-a_k)\}$
\State Target: $\log\pi(y,z)=\ell(\eta)+\alpha\sum_k\log w_k+\sum_k\log\{v_k(1-v_k)\}$, where $\ell(\eta)=\sum_i\log q_\eta(\hU_i)$ and $v_k=\operatorname{expit}(z_k)$
\For{$t=1,\ldots,T$}
\State Random-walk Metropolis updates of $y$, then of each $z_k$
\If{$t\equiv0 \pmod m$}
\State For each $k$, accept $\theta_k'\sim\mathrm U[a_k,b_k]$ with probability $\min\{1,e^{\ell(\eta')-\ell(\eta)}\}$
\EndIf
\State Block random-walk update of $(y,z)$
\State \textbf{if} $t\le T_0$ \textbf{then} adapt the proposal scales and block covariance \textbf{else} store $(w,\theta)$
\EndFor
\end{algorithmic}
\end{algorithm}

\subsection{Competitor and posterior summaries}\label{sec:summaries}

The competitor is multi-start maximum pseudo-likelihood (MPL), which maximizes the same pseudo-log-likelihood $\sum_i\log q_\eta(\hU_i)$ from ten random starting values, with the same pseudo-observations, mixture order and parameter bounds. Bayesian fits are summarized by the posterior predictive density $\overline q_n$ and by posterior means and 95\% credible intervals of the weights, the component Kendall's $\tau$ and the tail coefficients. Tail asymmetry is summarized by the posterior probability that one tail coefficient exceeds the other and, at finite levels, by the curves $\chi_L(t)$ and $\chi_U(t)$ in~\eqref{eq:chi}. Their empirical counterparts are the proportions $\#\{i:\hU_{i1}\le t,\hU_{i2}\le t\}/\#\{i:\hU_{i1}\le t\}$ and their upper-tail analogues; by Proposition~\ref{prop:functionals}, the posterior of $\chi_L(t)$ and $\chi_U(t)$ concentrates at the truth for each fixed $t$.

Models are compared out of sample. Each dataset is split once, within strata, into a training and an evaluation sample. The evaluation records are transformed by the training-sample empirical distribution functions, and each model is scored by the mean log posterior predictive copula density of the evaluation records, a strictly proper scoring rule \citep{Gneiting2007}, with paired standard errors for differences between models. For the tails, the evaluation records are re-ranked within strata, ties broken at random, and the observed $\widehat\chi_L(t)$ and $\widehat\chi_U(t)$ at $t\in\{0.05,0.10,0.20\}$ are compared with the predictions of the training fits.

\section{Simulation study}\label{sec:simulation}

\subsection{Design}

All experiments re-estimate the margins from each simulated sample. A fitted oracle uses the true $U_i$ but still estimates all mixture parameters; comparing it with the pseudo-posterior isolates the cost of estimating the margins. Density accuracy is measured by the Kullback--Leibler divergence $\KL(c_0,\widehat c\,)$, estimated from 16384 draws from $c_0$ shared across replicates and methods. For the Bayesian fits, $\widehat c$ is the posterior predictive density. Computing times refer to a laptop with an Intel Core i9-13980HX processor and 16 GB of memory, running Windows 11 with Python 3.11.5, NumPy 1.24.3 and SciPy 1.11.1. Each fit used a single thread, with up to twelve fits running in parallel.

\emph{Experiment 1} uses three two-component mixtures, each with weight $0.55$ on the first component: a bivariate Clayton--Gumbel mixture with $\theta=1.5$ and $\gamma=2$ (CG2), and five-dimensional Gaussian (N5) and Student $t_5$ (t5) mixtures with exchangeable correlations $0.10$ and $0.65$. The margins are lognormal, logistic, $t_4$, gamma and cubed normal. For each $n\in\{200,800,3200\}$ we draw 30 datasets per model. Four fits are compared: the pseudo-posterior with rank margins, the pseudo-posterior with Gaussian-kernel margins (Silverman's bandwidth), the oracle and rank-based MPL. The Bayesian fits use data augmentation with an additional block move (Supplementary Section~S2).

\emph{Experiment 2} uses the four-family bivariate mixture of Theorem~\ref{thm:identification}, with weights $0.4$, $0.3$, $0.2$ and $0.1$ on a Clayton, a Gumbel, a Frank and a Gaussian component. The weights are unequal so that the prior mean of the weights is not the truth. In the \emph{separated} scenario every component has Kendall's $\tau=0.6$ ($\theta_{\mathrm C}=3$, $\gamma=2.5$, $\theta_{\mathrm F}=7.93$, $\rho=0.81$). The components then differ in shape and tails, not in strength, and the tail coefficients are $\lambda_L=0.317$ and $\lambda_U=0.204$. In the \emph{near-independence} scenario every component has $\tau=0.1$ ($\theta_{\mathrm C}=0.22$, $\gamma=1.11$, $\theta_{\mathrm F}=0.91$, $\rho=0.16$), giving $\lambda_L=0.018$ and $\lambda_U=0.040$. For each $n\in\{500,2000,8000\}$ we draw 30 datasets per scenario and fit, with $K=4$, the rank-based pseudo-posterior (Algorithm~\ref{alg:sampler}; 8000 iterations, 2000 burn-in, thinning by six), the oracle and rank-based MPL. \emph{Experiment 3} reruns the first ten datasets of each cell of Experiment~2 with both samplers under identical settings.

\subsection{Experiment 1: two-component mixtures}

\begin{figure}[t]
\centering
\includegraphics[width=\textwidth]{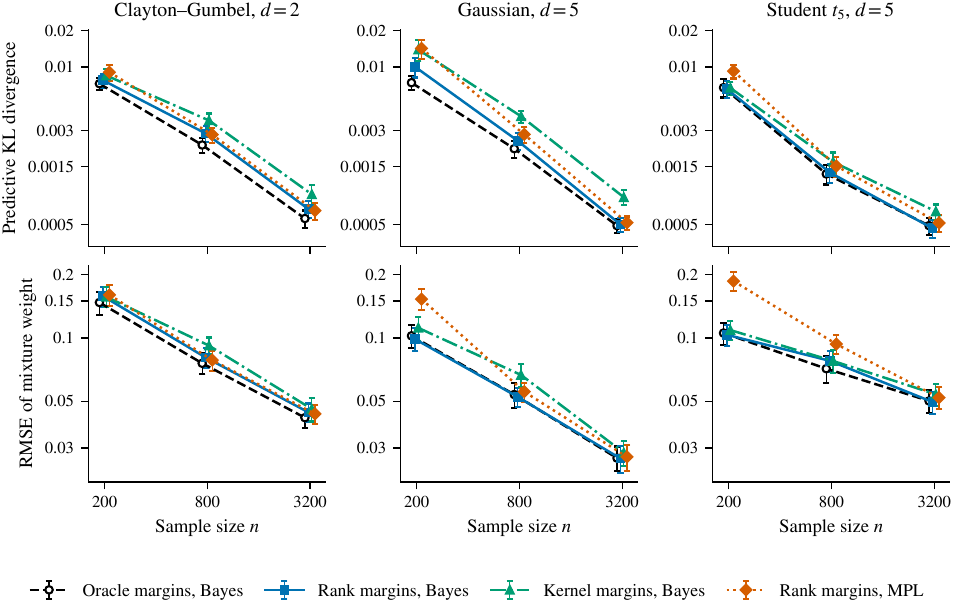}
\caption{Experiment 1. Predictive Kullback--Leibler divergence (top) and root mean squared error of the mixture weight (bottom), averaged over 30 replicates; bars show one standard error. Both axes are logarithmic; points are slightly offset horizontally.}
\label{fig:sim1}
\end{figure}

All errors decrease steadily with $n$ (Figure~\ref{fig:sim1}). The rank-based pseudo-posterior stays close to the oracle, so estimating the margins costs little once $n\ge800$. The pseudo-posterior and MPL differ mainly at $n=200$. There the posterior predictive density has a smaller divergence ($0.0101$ against $0.0143$ in N5, $0.0066$ against $0.0092$ in t5), and the posterior mean weight a markedly smaller error ($0.099$ against $0.153$ in N5, $0.104$ against $0.187$ in t5). With few observations the pseudo-likelihood is flat, and averaging over the posterior is more stable than committing to its maximizer. Part of the advantage in the weight, however, comes from shrinkage towards the prior mean $0.5$, which happens to be close to the true $0.55$. At $n=3200$ the two methods are indistinguishable. MPL is three to six times faster, at 0.2--1.3 seconds per fit with ten starts against 0.9--5.0 seconds for the sampler. Kernel margins are slightly less accurate than ranks. The gap is most visible for N5, whose cubed-normal margin has a density that is unbounded at zero and poorly followed by a smoothed estimate.

\subsection{Experiment 2: four-family mixtures}

\begin{figure}[t]
\centering
\includegraphics[width=\textwidth]{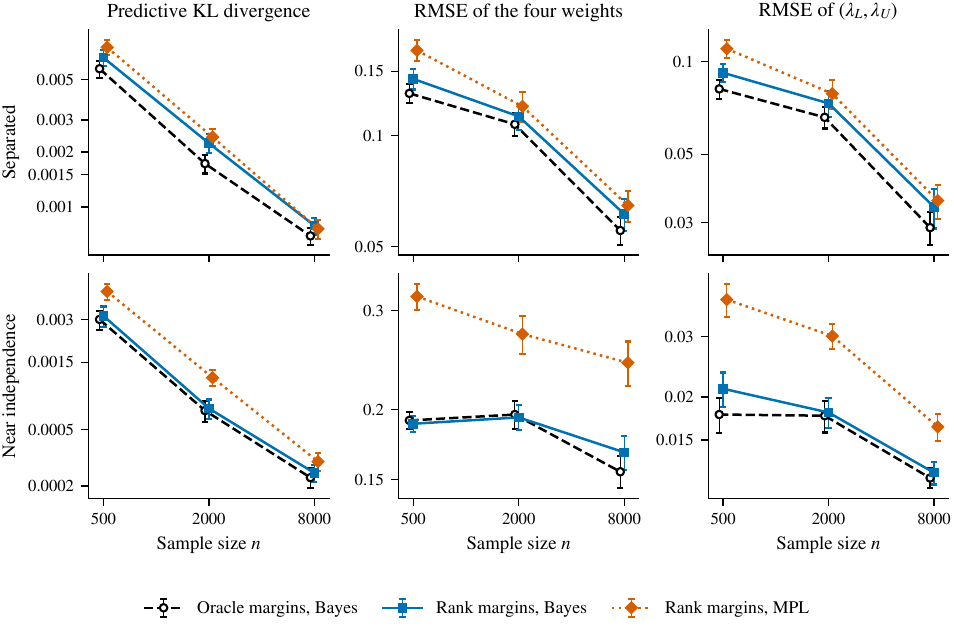}
\caption{Experiment 2, four-family mixture. Predictive Kullback--Leibler divergence (left), root mean squared error of the four weights (middle) and of the tail coefficients $(\lambda_L,\lambda_U)$ (right), averaged over 30 replicates, in the separated (top) and near-independence (bottom) scenarios; bars show one standard error.}
\label{fig:sim2}
\end{figure}

Experiment 2 (Figure~\ref{fig:sim2} and Table~\ref{tab:sim2}) separates what is learned about the density from what is learned about the components. In both scenarios the predictive divergence falls at the rate seen in Experiment 1; at $n=8000$ the rank-based pseudo-posterior is within 15\% of the oracle. With separated components the weights are learned too. Their root mean squared error falls from 0.14 to 0.06, and the 95\% credible intervals narrow from 0.46 to 0.18. Near independence the picture changes. The density is estimated even more accurately, but the weights hardly at all: their error stays between 0.17 and 0.19, and the credible intervals are still 0.64 wide at $n=8000$. These wide intervals reflect how little information the data carry for telling four nearly independent components apart. Identifiability guarantees eventual recovery (Corollary~\ref{cor:mixing}); how soon it comes depends on how distinct the components are. The tail coefficients depend only on the density, and their errors are below 0.04 at $n=8000$ in both scenarios.

For separated components at $n=8000$, MPL matches the pseudo-posterior. Elsewhere it is less reliable. At $n=500$ it returns an estimate on the boundary of the parameter set in 40\% of the separated and 83\% of the near-independent datasets. Near independence, its best starting points disagree substantially in at least 73\% of the datasets at every sample size, and its predictive divergence is larger by half or more for $n\le2000$. Where the pseudo-posterior reports wide intervals, MPL reports an essentially arbitrary point. The posterior has its price in computation, the sampler being four to nine times slower than MPL. An equal-weight version of the experiment gives the same picture (Supplementary Section~S3).

\begin{table}[t]
\centering\footnotesize
\setlength{\tabcolsep}{4pt}
\caption{Experiment 2, four-family mixture with $K=4$ and true weights $(0.4,0.3,0.2,0.1)$; 30 replicates per cell. KL: predictive Kullback--Leibler divergence ($\times10^{3}$). RMSE: root mean squared error of the four weights. Width: average width of 95\% credible intervals for the weights. Boundary: percentage of MPL fits with an estimate on the boundary of the parameter set. Unstable: percentage of MPL fits in which starts within one log-likelihood unit of the best differ by more than 0.1 in some component's Kendall's $\tau$. B, rank-based pseudo-posterior; M, rank-based MPL.}
\label{tab:sim2}
\begin{tabular}{@{}llccccccc@{}}
\toprule
& & \multicolumn{2}{c}{KL} & \multicolumn{2}{c}{RMSE} & Width & Boundary & Unstable\\
\cmidrule(lr){3-4}\cmidrule(lr){5-6}
Scenario & $n$ & B & M & B & M & B & M (\%) & M (\%)\\
\midrule
\multirow{3}{*}{Separated} & 500 & 6.59 & 7.51 & 0.143 & 0.171 & 0.46 & 40 & 73\\
 & 2000 & 2.24 & 2.42 & 0.113 & 0.120 & 0.32 & 13 & 17\\
 & 8000 & 0.79 & 0.76 & 0.061 & 0.065 & 0.18 & 3 & 10\\
\midrule
\multirow{3}{*}{Near independence} & 500 & 3.18 & 4.74 & 0.189 & 0.318 & 0.79 & 83 & 80\\
 & 2000 & 0.71 & 1.16 & 0.194 & 0.272 & 0.75 & 63 & 83\\
 & 8000 & 0.25 & 0.30 & 0.168 & 0.242 & 0.64 & 43 & 73\\
\bottomrule
\end{tabular}
\end{table}

\subsection{Experiment 3: sampler efficiency}\label{sec:exp3}

Table~\ref{tab:sampler} compares the marginal sampler (Algorithm~\ref{alg:sampler}) with data augmentation. Per iteration, the marginal sampler takes 1.7 to 2.3 times as long. In return it yields two to four times the effective sample size for the weights near independence, and two to three times that for the tail coefficients in both scenarios. Per second of computation, the two samplers are about equally efficient for separated components; near independence, where mixing is hardest, the marginal sampler is 1.3 to 1.9 times as efficient. Both target the same posterior, and their posterior means of the tail coefficients differ by at most 0.026 in every replicate. Under either sampler the weights mix slowly near independence: from 1000 retained draws at $n=8000$, the median effective sample size is 18 for data augmentation and 65 for the marginal sampler. This reflects the flat posterior seen in Experiment~2, not a defect of either sampler, and the applications therefore use several chains and longer runs.

\begin{table}[t]
\centering\footnotesize
\setlength{\tabcolsep}{4pt}
\caption{Experiment 3: data-augmentation (DA) and marginal (MG) samplers for the four-family pseudo-posterior of Experiment 2, 10 replicates per cell, 8000 iterations each. Time: mean wall-clock seconds. ESS: median over replicates of the smallest effective sample size among the four weights, or between the two tail coefficients. ESS/s ratio: median over replicates of the MG-to-DA ratio of ESS per second.}
\label{tab:sampler}
\begin{tabular}{@{}llcccccccc@{}}
\toprule
& & \multicolumn{2}{c}{Time (s)} & \multicolumn{2}{c}{ESS, weights} & \multicolumn{2}{c}{ESS, $(\lambda_L,\lambda_U)$} & \multicolumn{2}{c}{ESS/s ratio}\\
\cmidrule(lr){3-4}\cmidrule(lr){5-6}\cmidrule(lr){7-8}\cmidrule(l){9-10}
Scenario & $n$ & DA & MG & DA & MG & DA & MG & Weights & Tails\\
\midrule
\multirow{3}{*}{Separated} & 500 & 12 & 27 & 74 & 146 & 87 & 205 & 0.9 & 1.0\\
 & 2000 & 28 & 52 & 61 & 92 & 64 & 136 & 0.8 & 1.1\\
 & 8000 & 86 & 142 & 52 & 83 & 76 & 192 & 1.0 & 1.9\\
\midrule
\multirow{3}{*}{Near independence} & 500 & 11 & 22 & 68 & 179 & 177 & 449 & 1.4 & 1.3\\
 & 2000 & 27 & 50 & 27 & 99 & 114 & 282 & 1.9 & 1.4\\
 & 8000 & 84 & 142 & 18 & 65 & 66 & 181 & 1.7 & 1.4\\
\bottomrule
\end{tabular}
\end{table}

\section{Applications to physical fitness data}\label{sec:applications}

\subsection{Data and standardization}\label{sec:data}

The first dataset contains the national physical fitness test records of 8772 students (6368 women, 2404 men) in their first three years at four universities in Jiangxi Province, China \citep{Sun2024,SunData2024}. We use three tests: the 50-m sprint, with times reversed so that larger values mean faster running; the standing long jump, a standard field measure of lower-limb explosive power \citep{CastroPinero2010}; and the sit-and-reach test of flexibility. The national standard scores men and women on separate scales, and the margins are accordingly standardized within sex.

The second dataset combines the 2011--2012 and 2013--2014 cycles of NHANES \citep{NHANES}. Muscular strength is measured by combined grip strength, the sum of the largest readings of the two hands on a handheld dynamometer. Daily activity comes from a wrist-worn accelerometer worn for seven days and is summarized in monitor-independent movement summary (MIMS) units \citep{John2019,Belcher2021}. We keep adults aged 20 to 80 years (age is top-coded at 80) who completed the grip test and had at least four valid days, a day being valid if at least 95\% of its 1440 minutes were classified as wear. Activity is the mean daily MIMS total over valid days. This leaves 5336 adults (3081 women, 2255 men; Supplementary Section~S4). Grip strength declines with age and differs between the sexes \citep{Dodds2014}, and low strength is judged against sex-specific norms \citep{CruzJentoft2019}. We therefore standardize the margins within six strata formed by sex and age group (20--39, 40--59 and 60--80 years), so that ``weak'' or ``inactive'' always means weak or inactive relative to adults of the same sex and age group. Table~\ref{tab:data} gives marginal summaries.

In both datasets sex is the binary variable recorded by the data provider (in NHANES, sex as reported at the interview); gender identity was not recorded.

\begin{table}[t]
\centering\small
\setlength{\tabcolsep}{4pt}
\caption{Marginal summaries by sex: median (lower quartile, upper quartile). MIMS, monitor-independent movement summary units from the wrist accelerometer, averaged over valid days.}
\label{tab:data}
\begin{tabular}{@{}llcc@{}}
\toprule
Data & Measurement & Women & Men\\
\midrule
\multirow{4}{*}{\shortstack[l]{Student fitness\\tests (Jiangxi)}} & Students & 6368 & 2404\\
& 50-m sprint time (s) & 9.1 (8.7, 9.6) & 7.4 (7.1, 7.8)\\
& Standing long jump (cm) & 170 (161, 181) & 228 (216, 240)\\
& Sit-and-reach (cm) & 20 (16, 24) & 15 (11, 19)\\
\midrule
\multirow{4}{*}{\shortstack[l]{Adults\\(NHANES)}} & Participants & 3081 & 2255\\
& Age (years) & 49 (35, 63) & 47 (33, 62)\\
& Combined grip strength (kg) & 55.6 (47.8, 63.2) & 88.0 (75.3, 99.0)\\
& Daily activity ($10^3$ MIMS) & 14.7 (12.5, 16.8) & 14.5 (12.1, 17.1)\\
\bottomrule
\end{tabular}
\end{table}

Three features of the data limit a literal reading of the theory. First, sprint times are recorded to 0.1 seconds, jumps to the centimeter and grip strength to 0.1 kg, so ties occur. We use average ranks and repeat every reported fit with ties broken at random. Second, both samples come from multistage cluster designs. The university of each student is not recorded, and we do not use the NHANES survey weights, so the uncertainty statements describe the sampled participants rather than a target population. Third, both studies are cross-sectional; the associations below are not effects of training or of aging within a person. All four-family fits use the priors of Section~\ref{sec:priors} and two to four chains of 10000 to 12000 iterations of Algorithm~\ref{alg:sampler}. Every reported fit has rank-normalized split $\widehat R$ at most 1.02 \citep{Vehtari2021}, and every reported tail coefficient an effective sample size above 400 (Supplementary Section~S5).

\subsection{Students: speed, explosive power and flexibility}

\begin{figure}[t]
\centering
\includegraphics[width=\textwidth]{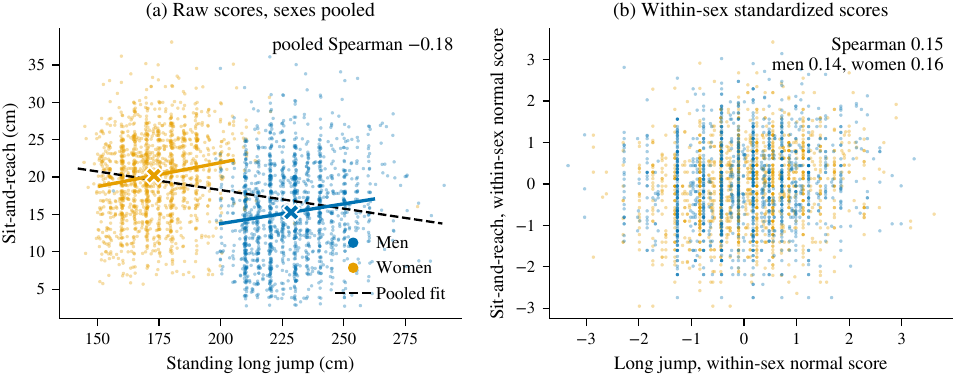}
\caption{Standing long jump and sit-and-reach in 8772 students. (a) Raw scores with ordinary least-squares lines for each sex (solid) and for the pooled sample (dashed); 1500 randomly chosen points per sex, slightly jittered. (b) Normal scores of within-sex ranks.}
\label{fig:reversal}
\end{figure}

Men jump farther and women reach farther (Table~\ref{tab:data}). Pooling the sexes therefore produces a negative association between jump and sit-and-reach (Spearman correlation $-0.18$), which would suggest a trade-off between explosive power and flexibility (Figure~\ref{fig:reversal}a). Within each sex the association is positive: 0.14 in men, 0.16 in women, and 0.15 for the pooled within-sex standardized scores (Figure~\ref{fig:reversal}b). The apparent trade-off is a composition effect of the kind described by \citet{Simpson1951}. Standardization also reduces the speed--power association from 0.76, inflated by the sex difference in both tests, to 0.44.

We first fitted three-dimensional models to the standardized speed, power and flexibility scores, with 3000 training and 5772 evaluation records split within sex and grade. Two- and three-component Gaussian mixtures improve slightly on a single Gaussian copula. A Student $t_5$ copula, which imposes tail dependence on every pair, is clearly worse (Table~\ref{tab:heldout}). In the two-component fit, the speed--flexibility and power--flexibility correlations stay below 0.2 in both components, whereas the speed--power correlation is 0.27 in one and 0.63 in the other (Supplementary Section~S5). Flexibility is thus an almost separate dimension of fitness; the heterogeneity lies in how strongly speed and power are coupled.

\begin{table}[t]
\centering\small
\caption{Mean held-out log copula density per subject, with the difference from a single Gaussian copula and its paired standard error. Students: 3000 training and 5772 evaluation records. Adults: 2499 training and 2837 evaluation records.}
\label{tab:heldout}
\begin{tabular}{@{}llcc@{}}
\toprule
Data and variables & Model & Mean log density & Difference (SE)\\
\midrule
\multirow{4}{*}{\shortstack[l]{Students: speed,\\power, flexibility}} & Gaussian, $K=1$ & 0.1162 & ---\\
 & Gaussian, $K=2$ & 0.1173 & 0.0011 (0.0003)\\
 & Gaussian, $K=3$ & 0.1174 & 0.0013 (0.0004)\\
 & Student $t_5$, $K=1$ & 0.0896 & $-$0.0266 (0.0037)\\
\midrule
\multirow{3}{*}{\shortstack[l]{Students: speed,\\power}} & Gaussian, $K=1$ & 0.1029 & ---\\
 & Gaussian, $K=2$ & 0.1048 & 0.0019 (0.0005)\\
 & Four-family, $K=4$ & 0.1079 & 0.0050 (0.0011)\\
\midrule
\multirow{3}{*}{\shortstack[l]{Adults: grip,\\activity}} & Gaussian, $K=1$ & 0.0015 & ---\\
 & Gaussian, $K=2$ & 0.0025 & 0.0010 (0.0012)\\
 & Four-family, $K=4$ & 0.0086 & 0.0071 (0.0019)\\
\bottomrule
\end{tabular}
\end{table}

A Gaussian mixture is radially symmetric. It cannot say whether the strong coupling sits at the top or at the bottom of the distribution, but the four-family mixture can. Fitted to the standardized speed and power scores, it improves on a single Gaussian copula by $0.0050$ per student on the evaluation records (standard error $0.0011$), more than twice the gain of a two-component Gaussian mixture (Table~\ref{tab:heldout}). In the full sample the Clayton component receives almost no weight, 0.01 (0.00 to 0.05), consistent with an absent Clayton component (Corollary~\ref{cor:fourfamily}(a)); its Kendall's $\tau$ is barely constrained by the data and is not interpreted. The upper tail-dependence coefficient is $\lambda_U=0.10$ (95\% interval 0.06 to 0.15), about ten times the lower one (Table~\ref{tab:tail}), and every posterior draw has $\lambda_U>\lambda_L$. Under the model these coefficients are consistently estimated whether or not the weights are (Proposition~\ref{prop:functionals}).

\begin{table}[!htbp]
\centering\footnotesize
\setlength{\tabcolsep}{3pt}
\caption{Four-family mixture fits: posterior means and 95\% credible intervals. Top: component weights and component Kendall's $\tau$ in the full samples (average ranks); the $\tau$ of a component with posterior weight near zero is barely constrained by the data and is not interpreted. Bottom: lower and upper tail-dependence coefficients and the posterior probability of the dominant tail ordering, by subsample. Ranks are computed within sex (students) or within sex and age group (adults).}
\label{tab:tail}
\begin{tabular}{@{}lcccc@{}}
\toprule
& \multicolumn{2}{c}{Students} & \multicolumn{2}{c}{Adults}\\
\cmidrule(lr){2-3}\cmidrule(l){4-5}
Component & Weight & Kendall's $\tau$ & Weight & Kendall's $\tau$\\
\midrule
Clayton & 0.01 (0.00, 0.05) & 0.59 (0.12, 0.85) & 0.63 (0.44, 0.82) & 0.14 (0.09, 0.19)\\
Gumbel & 0.35 (0.21, 0.48) & 0.23 (0.17, 0.32) & 0.01 (0.00, 0.08) & 0.58 (0.01, 0.85)\\
Frank & 0.50 (0.33, 0.66) & 0.35 (0.29, 0.43) & 0.02 (0.00, 0.16) & 0.52 (0.01, 0.87)\\
Gaussian & 0.14 (0.00, 0.35) & 0.29 ($-$0.10, 0.50) & 0.33 (0.15, 0.53) & $-$0.16 ($-$0.28, $-$0.08)\\
\bottomrule
\end{tabular}

\vspace{0.8em}
\begin{tabular}{@{}lcccc@{}}
\toprule
\multicolumn{5}{@{}l}{\emph{Students: sprint speed and standing long jump}}\\
& $n$ & $\lambda_L$ & $\lambda_U$ & $\Pr(\lambda_U>\lambda_L)$\\
\cmidrule(l){2-5}
All, average ranks & 8772 & 0.01 (0.00, 0.03) & 0.10 (0.06, 0.15) & 1.00\\
All, random ties (1) & 8772 & 0.01 (0.00, 0.03) & 0.10 (0.06, 0.15) & 1.00\\
All, random ties (2) & 8772 & 0.01 (0.00, 0.03) & 0.11 (0.06, 0.15) & 1.00\\
Women & 6368 & 0.01 (0.00, 0.03) & 0.12 (0.06, 0.17) & 1.00\\
Men & 2404 & 0.02 (0.00, 0.05) & 0.05 (0.00, 0.12) & 0.86\\
First year & 2809 & 0.01 (0.00, 0.04) & 0.14 (0.07, 0.21) & 1.00\\
Second year & 2904 & 0.02 (0.00, 0.08) & 0.03 (0.00, 0.09) & 0.52\\
Third year & 3059 & 0.03 (0.00, 0.06) & 0.13 (0.06, 0.21) & 0.99\\
\midrule
\multicolumn{5}{@{}l}{\emph{Adults: grip strength and daily activity}}\\
& $n$ & $\lambda_L$ & $\lambda_U$ & $\Pr(\lambda_L>\lambda_U)$\\
\cmidrule(l){2-5}
All, average ranks & 5336 & 0.07 (0.03, 0.12) & 0.00 (0.00, 0.01) & 1.00\\
All, random ties (1) & 5336 & 0.07 (0.03, 0.11) & 0.00 (0.00, 0.01) & 1.00\\
All, random ties (2) & 5336 & 0.07 (0.03, 0.12) & 0.00 (0.00, 0.01) & 1.00\\
Age 20--39 & 1804 & 0.02 (0.00, 0.07) & 0.01 (0.00, 0.03) & 0.63\\
Age 40--59 & 1909 & 0.06 (0.01, 0.11) & 0.01 (0.00, 0.03) & 0.95\\
Age 60--80 & 1623 & 0.16 (0.08, 0.23) & 0.01 (0.00, 0.03) & 1.00\\
\bottomrule
\end{tabular}
\end{table}

\begin{figure}[t]
\centering
\includegraphics[width=\textwidth]{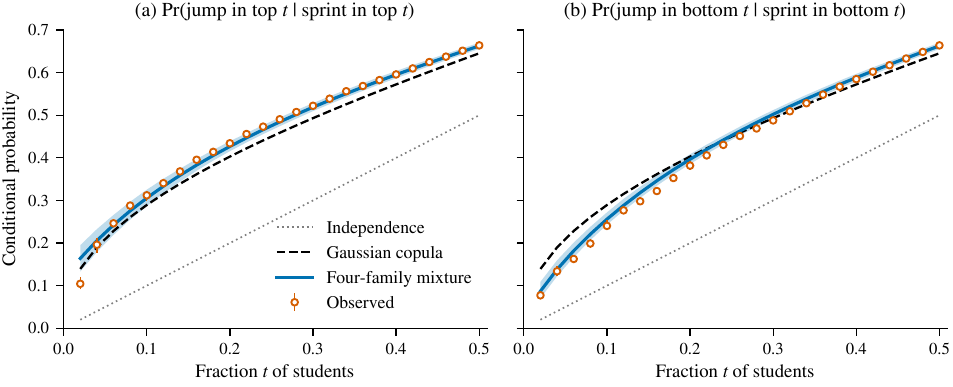}
\caption{Co-occurrence of extreme sprint and jump performance among 8772 students, after within-sex standardization. Observed proportions (circles; bars show the range over 20 random tie-breakings) and posterior means for the four-family mixture (solid line, with 95\% band) and a single Gaussian copula (dashed). The dotted line is the value $t$ under independence.}
\label{fig:tails}
\end{figure}

In terms of proportions of students (Figure~\ref{fig:tails}), the result reads as follows. Among the fastest 10\% of sprinters within their sex, 31\% are also in the top 10\% for the standing long jump, three times the rate expected under independence. Among the slowest 10\%, only 24\% are in the bottom 10\% for the jump. The four-family mixture reproduces both curves. A Gaussian copula, being symmetric, overstates co-weakness and understates co-excellence. Breaking ties at random leaves the asymmetry unchanged. It is clear in women ($\lambda_U=0.12$, $\Pr(\lambda_U>\lambda_L)=1.00$) and weaker in the smaller sample of men (0.05 and 0.86). We also fitted the model separately by year of study, with margins standardized within sex and year. The asymmetry appears in first-year and third-year students but not in second-year students, for whom $\lambda_U=0.03$, $\lambda_L=0.02$ and $\Pr(\lambda_U>\lambda_L)=0.52$. The empirical proportions agree. At the 10\% level, co-excellence and co-weakness are 0.29 and 0.21 in the first year and 0.35 and 0.23 in the third, against 0.30 and 0.29 in the second year, where co-weakness is more common. The data do not identify the reason, since the cohorts differ in more than their year of study.

In practice, the two tests overlap more at the top of the distribution than at the bottom. A student in the top 10\% for one explosive test has about a 31\% chance of also being in the top 10\% for the other, three times the rate under independence, consistent with a shared neuromuscular basis for sprinting and jumping. Poor results are more often specific to one test. At the bottom of the distribution one test is therefore a weaker proxy for the other than at the top, and a student who needs support in one may be missed by screening with the other alone. Flexibility should be assessed in its own right; within sex, no trade-off with power need be expected.

\subsection{Adults: grip strength and daily activity}

Judged by a correlation coefficient, grip strength and daily activity are almost unrelated. The Spearman correlation of the standardized scores is $0.05$ overall, and $0.03$, $-0.05$ and $0.21$ in the three age groups. The copula fits point elsewhere. On the evaluation records (2499 training and 2837 evaluation participants, split within strata), the four-family mixture improves the mean log copula density over a single Gaussian copula by $0.0071$ per participant (standard error $0.0019$). A two-component Gaussian mixture gains only $0.0010$ ($0.0012$; Table~\ref{tab:heldout}). Extra flexibility without asymmetry thus buys almost nothing, which points to asymmetry, not heterogeneity as such, as the source of the improvement.

\begin{figure}[t]
\centering
\includegraphics[width=\textwidth]{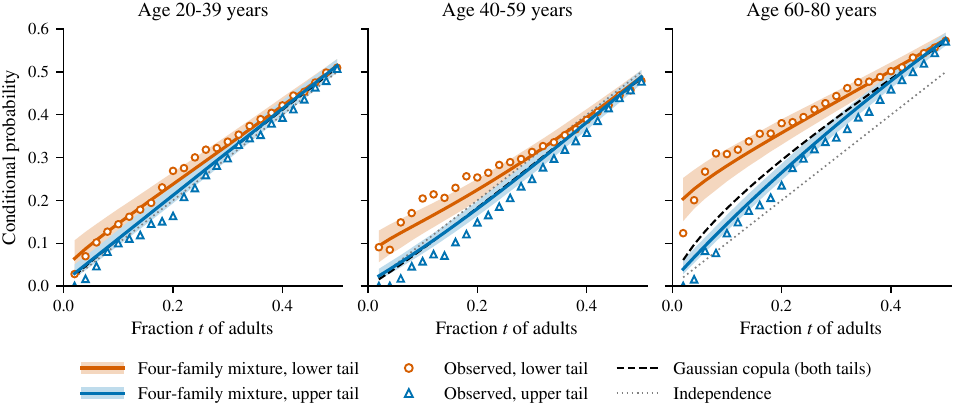}
\caption{Co-occurrence of low and of high grip strength and daily activity among 5336 adults, by age group, after standardization within sex and age group. Lower tail: proportion of adults in the bottom fraction $t$ of daily activity among those in the bottom fraction $t$ of grip strength, $\chi_L(t)$ in~\eqref{eq:chi}; upper tail: the same for the top fraction, $\chi_U(t)$. Observed proportions (points) and posterior means of the four-family mixture fitted to each age group (solid lines, with 95\% bands); the Gaussian copula fitted to the same data gives the same value for both tails (dashed); the dotted line is the value $t$ under independence.}
\label{fig:nhanes}
\end{figure}

The full-sample fit puts weight $0.63$ on a Clayton component with $\tau=0.14$ and $0.33$ on a Gaussian component with \emph{negative} $\tau=-0.16$ (Table~\ref{tab:tail}). The Gumbel and Frank components receive weights of 0.01 and 0.02, and their wide $\tau$ intervals carry no information. Weakness and inactivity thus go together, while elsewhere in the distribution the association is slightly negative; in a correlation coefficient the two patterns nearly cancel. The lower tail-dependence coefficient is $\lambda_L=0.07$ (95\% interval 0.03 to 0.12), the upper one essentially zero, and $\Pr(\lambda_L>\lambda_U)=0.999$. Random tie-breaking changes nothing. The asymmetry grows with age. Fitted within age groups, $\lambda_L$ rises from $0.02$ at 20--39 years to $0.06$ at 40--59 years and $0.16$ (0.08 to 0.23) at 60--80 years, with posterior probabilities of $\lambda_L>\lambda_U$ of $0.63$, $0.95$ and $1.00$.

At finite levels the pattern is easy to read (Figure~\ref{fig:nhanes}). Consider adults aged 60 to 80 years whose grip strength is in the weakest 10\% for their sex and age group. Of these, 31\% are also among the least active 10\%, three times the rate under independence and well above the 18\% implied by a Gaussian copula fitted to the same data. Among the strongest 10\%, only 12\% are among the most active 10\%, close to independence. In adults under 40 the two tails differ much less (0.14 and 0.10). The four-family mixture tracks the observed proportions closely in the oldest group and reproduces the ordering of the tails in the others. The Gaussian copula averages the two tails and misrepresents both.

Weakness and low physical activity are two of the five criteria of the frailty phenotype \citep{Fried2001}, and grip strength is a cheap, prognostic measurement \citep{Leong2015}. Our results add a distributional statement. In older adults a weak grip marks a group in which very low daily activity is concentrated, whereas a strong grip says little about high activity. A correlation of $0.05$, or even $0.21$ in the oldest group, would suggest that grip strength says little about activity. The tail analysis shows that it carries information precisely in the lower tail, where screening takes place. The data are cross-sectional and unweighted, and chronic disease or disability that limits both strength and movement may produce the co-occurrence. The result therefore describes in whom weakness and inactivity co-occur, not why.

\subsection{Held-out validation of tail probabilities}

\begin{figure}[t]
\centering
\includegraphics[width=\textwidth]{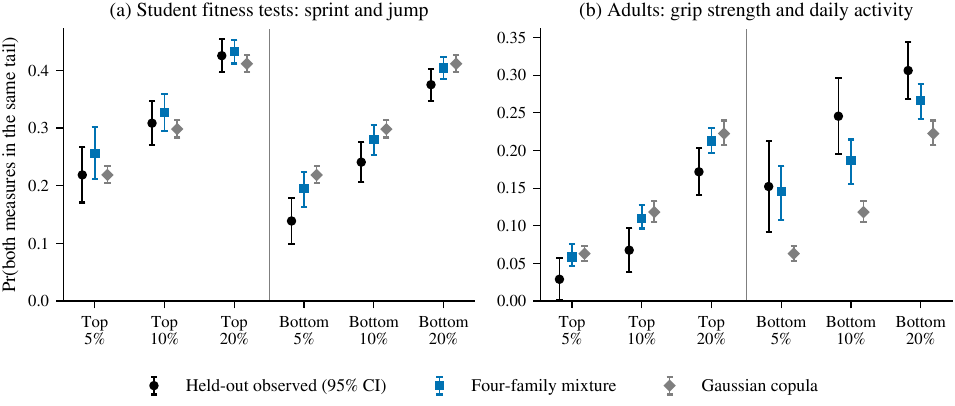}
\caption{Held-out validation of tail co-occurrence. For the evaluation records, the observed proportion of records in the top (left of the vertical line) or bottom (right) fraction $t$ of the second measure among those in the same fraction of the first (circles, with 95\% intervals), against posterior predictions from the training fits of the four-family mixture (squares) and a Gaussian copula (diamonds), with 95\% credible intervals. (a) Students, sprint speed and standing long jump, 5772 evaluation records. (b) Adults, grip strength and daily activity, 2837 evaluation records.}
\label{fig:validation}
\end{figure}

Figure~\ref{fig:validation} checks the tail predictions of the training fits against the evaluation records. For the adults, the four-family mixture predicts that 14.6\% of the weakest 5\% are also among the least active 5\%; the observed proportion is 15.2\%, and the Gaussian copula predicts 6.3\%. At the 10\% and 20\% levels the four-family prediction falls between the observation and the Gaussian prediction, and the Gaussian prediction for the lower tail lies outside the 95\% interval of the observed proportion at every level. In the upper tail both models overstate co-occurrence, the four-family mixture slightly less. For the students, both models predict co-excellence well. For co-weakness, however, the Gaussian copula predicts the same value as for co-excellence, $0.22$ at the 5\% level, where $0.14$ is observed, while the four-family mixture predicts $0.20$. In both datasets the four-family mixture orders the tails correctly, which the Gaussian copula cannot do.

\section{Discussion}\label{sec:discussion}

Replacing the margins by ranks or kernel estimates in a Bayesian copula likelihood gives a consistent estimate of the copula density, provided the log densities obey a boundary envelope. Gaussian, Student, Clayton, Gumbel and Frank copulas and their finite mixtures all do. Density consistency, together with continuity in the parameter, delivers tail coefficients, conditional tail probabilities and rank correlations. Identifiability, established by Theorem~\ref{thm:identification} for the Clayton--Gumbel--Frank--Gaussian dictionary, matters only for statements about the components. The simulations show why the distinction is worth making. Near independence the density and the tail coefficients are estimated accurately, while the weights are barely learned even at $n=8000$. Maximum pseudo-likelihood then often returns degenerate or unstable estimates, whereas the wide credible intervals of the pseudo-posterior make the lack of information visible. In that regime the marginal sampler is the more efficient choice.

In the applications, the shape of dependence matters more than its strength. Sprinting and jumping are linked mainly among the best students. Grip strength and daily activity are linked almost only among the weakest and least active adults, and most strongly in old age. In both cases a symmetric model with the right overall correlation gives the wrong answer in at least one tail, and held-out data confirm the asymmetry.

Several limitations point to further work. The theory assumes continuous margins, whereas fitness tests are recorded on coarse grids. A treatment of ties in the spirit of \citet{Genest2007} would be valuable, although our conclusions did not change under random tie-breaking. Survey weights could enter through a weighted pseudo-log-likelihood, but the theory would then have to account for the design. The dimension and the number of components are fixed, and the parameter sets compact. Two questions remain open for this dictionary: rates of contraction, which require stronger identifiability conditions \citep{Heinrich2018} and would quantify how fast absent families lose their weight, and the number of components within a family \citep{Rousseau2011}. Theorem~\ref{thm:identification} does not cover Student components with unknown, component-specific degrees of freedom. Finally, the pseudo-posterior ignores the uncertainty of the estimated margins, so its credible intervals need not have nominal frequentist coverage \citep{Miller2021,Smith2025}. Calibration in the spirit of cut posteriors is a natural next step.

\section*{Data and code availability}
The NHANES files are in the public domain and available from the National Center for Health Statistics \citep{NHANES}. The Jiangxi fitness data are the supporting dataset of \citet{Sun2024}, available under a CC BY 4.0 license \citep{SunData2024}. The wearable-device data analysed in the Supplementary Material are available under a CC0 license \citep{BroaddusData2018}. Python code that prepares all datasets and reproduces all simulations, analyses, tables and figures is provided as supplementary material.

\section*{CRediT authorship contribution statement}
\textbf{Yujian Liu:} Conceptualization, Methodology, Formal analysis, Software, Data curation, Visualization, Writing -- original draft, Writing -- review \& editing.

\section*{Declaration of competing interest}
The author declares no competing interests.

\section*{Funding}
This work was supported by a research project of the School of Economics and Management, Shanghai University of Sport.

\section*{Declaration of generative AI and AI-assisted technologies in the manuscript preparation process}
During the preparation of this work the author used generative AI tools to assist with language editing and proof checking. The author reviewed and edited the content as needed and takes full responsibility for the content of the published article.

\section*{Acknowledgements}
The author thanks the National Center for Health Statistics and the creators of the Jiangxi fitness dataset for making their data openly available.


\clearpage
\setcounter{section}{0}\setcounter{table}{0}\setcounter{figure}{0}\setcounter{equation}{0}
\renewcommand{\thesection}{S\arabic{section}}\renewcommand{\theHsection}{S\arabic{section}}
\renewcommand{\thetable}{S\arabic{table}}\renewcommand{\theHtable}{S\arabic{table}}
\renewcommand{\thefigure}{S\arabic{figure}}\renewcommand{\theHfigure}{S\arabic{figure}}
\renewcommand{\theequation}{S\arabic{equation}}\renewcommand{\theHequation}{S\arabic{equation}}
\begin{center}
{\Large\bfseries Supplementary Material\par}
\end{center}
\vspace{0.5em}
Section~\ref{sec:s-proofs} contains the proofs of the results in Section~3 of the main article; Section~\ref{sec:s-computation} describes the priors, samplers, competitor and accuracy measures; Section~\ref{sec:s-simulation} gives additional simulation results; Section~\ref{sec:s-data} describes the data and their preparation; Section~\ref{sec:s-applications} gives additional application results and convergence diagnostics; and Section~\ref{sec:s-running} reports an application to wearable-device data. Equation, theorem, algorithm and section numbers without the prefix S refer to the main article.

\section{Proofs}\label{sec:s-proofs}

The notation is that of Sections~2 and~3 of the main article.
Throughout, $\ell_\eta=\log q_\eta$, $b(t)=-\log\{t(1-t)\}$ and $E(u)=A+B\sum_jb(u_j)$. Almost-sure statements refer to the sampling distribution of $X_1,X_2,\ldots$; since the margins are continuous, ties among the $X_{ij}$ have probability zero. Suprema over $\eta$ are measurable because $H$ is a compact metric space, hence separable, and all functions of $\eta$ involved are continuous.

\subsection{Marginal estimates}\label{app:margins}

\begin{lemma}\label{lem:margins}
For scaled ranks, and for the kernel estimates~\eqref{eq:kernel} with $h_{nj}\to0$ almost surely, almost surely $\hU_i\in(0,1)^d$ for all $i$ and
\begin{equation}\label{eq:transfer-input}
\max_{1\le i\le n}\|\hU_i-U_i\|_\infty\to0,\qquad \limsup_{n\to\infty}\frac1n\sum_{i=1}^nE(\hU_i)^2<\infty .
\end{equation}
\end{lemma}

\begin{proof}
\emph{Rank grid.} Since $b(t)^2\le2(\log t)^2+2\{\log(1-t)\}^2$ and $t\mapsto(\log t)^2$ decreases on $(0,1)$, comparison with an integral gives $\sum_{a=1}^n\{\log(a/(n+1))\}^2\le(n+1)\int_0^1(\log t)^2dt=2(n+1)$, and the reflected sum obeys the same bound. Hence
\begin{equation}\label{eq:rank-moment}
\frac1n\sum_{a=1}^nb\!\left(\frac a{n+1}\right)^2\le\frac{8(n+1)}n\le16\qquad(n\ge1).
\end{equation}

\emph{Ranks.} Let $F_{nj}$ be the empirical distribution function of $X_{1j},\ldots,X_{nj}$, so that $R_{ij}=nF_{nj}(X_{ij})$. Then $|R_{ij}/(n+1)-F_j(X_{ij})|\le\|F_{nj}-F_j\|_\infty+1/(n+1)$, which tends to zero uniformly in $i$ by the Glivenko--Cantelli theorem. For each $j$ the ranks $R_{1j},\ldots,R_{nj}$ are a permutation of $1,\ldots,n$, so~\eqref{eq:rank-moment} and $E(u)^2\le(d+1)\{A^2+B^2\sum_jb(u_j)^2\}$ give $n^{-1}\sum_iE(\hU_i)^2\le(d+1)(A^2+16dB^2)$.

\emph{Kernel estimates.} Since $\widehat F_j(x)=\int F_{nj}(x-h_{nj}z)\,d\mathcal K_j(z)$,
\[
\|\widehat F_j-F_j\|_\infty\le\|F_{nj}-F_j\|_\infty+\int\sup_x|F_j(x-h_{nj}z)-F_j(x)|\,d\mathcal K_j(z).
\]
On the event $\{h_{nj}\to0\}$ the integrand tends to zero for every $z$, because a continuous distribution function is uniformly continuous, and it is bounded by one; by dominated convergence $\max_i|\hU_{ij}-U_{ij}|\le\|\widehat F_j-F_j\|_\infty\to0$. For the moment bound, let $X_{(1)j}<\cdots<X_{(n)j}$ be the order statistics and $\alpha_j=\mathcal K_j(0)$. As $\mathcal K_j$ is nondecreasing, each of the $a$ observations $X_{lj}\le X_{(a)j}$ contributes at least $\alpha_j/n$ to $\widehat F_j(X_{(a)j})$, and each of the $n-a+1$ observations $X_{lj}\ge X_{(a)j}$ contributes at least $(1-\alpha_j)/n$ to $1-\widehat F_j(X_{(a)j})$. Therefore
\[
\widehat F_j(X_{(a)j})\ge\alpha_j\frac an,\qquad 1-\widehat F_j(X_{(a)j})\ge(1-\alpha_j)\frac{n-a+1}n,
\]
so that $\hU_{ij}\in(0,1)$ and $b\{\widehat F_j(X_{(a)j})\}\le b\{a/(n+1)\}-\log\{\alpha_j(1-\alpha_j)\}$. Together with~\eqref{eq:rank-moment} this gives the moment bound.
\end{proof}

\subsection{Uniform transfer and proof of Theorem~\ref{thm:consistency}(a)}\label{app:transfer}

\begin{lemma}\label{lem:transfer}
Under Assumption~\ref{ass:envelope}, pseudo-observations satisfying~\eqref{eq:transfer-input} almost surely satisfy
\[
\sup_{\eta\in H}\left|\frac1n\sum_{i=1}^n\{\ell_\eta(\hU_i)-\ell_\eta(U_i)\}\right|\to0 .
\]
\end{lemma}

\begin{proof}
Fix $\delta\in(0,1/4)$ and let $G_\delta=[\delta,1-\delta]^d$. By Assumption~\ref{ass:envelope}, $(\eta,u)\mapsto\ell_\eta(u)$ is uniformly continuous on the compact set $H\times G_\delta$; let $\omega_\delta$ be its modulus of continuity in $u$, uniform in $\eta$. Split the indices into $I_n=\{i:U_i\in G_\delta,\ \hU_i\in G_\delta\}$ and its complement $B_n$. Then
\[
\frac1n\sum_{i\in I_n}\sup_\eta|\ell_\eta(\hU_i)-\ell_\eta(U_i)|\le\omega_\delta\Bigl(\max_i\|\hU_i-U_i\|_\infty\Bigr)\to0 .
\]
Once $\max_i\|\hU_i-U_i\|_\infty<\delta$, every $i\in B_n$ has $U_i\notin G_{2\delta}$, so by the strong law and the uniformity of the margins of $U$, $\limsup_n|B_n|/n\le\Pr(U\notin G_{2\delta})\le4d\delta$. By~\eqref{eq:envelope} and the Cauchy--Schwarz inequality,
\[
\frac1n\sum_{i\in B_n}\sup_\eta|\ell_\eta(\hU_i)-\ell_\eta(U_i)|\le\left(\frac{|B_n|}n\right)^{1/2}\left[\frac1n\sum_{i=1}^n\{E(\hU_i)+E(U_i)\}^2\right]^{1/2}.
\]
The second factor has an almost surely finite limit superior $\Gamma$, not depending on $\delta$, by~\eqref{eq:transfer-input} and the strong law applied to $E(U_i)^2$, whose mean is finite because each $U_j$ is uniform. The limit superior of the left side of the lemma is therefore at most $\Gamma(4d\delta)^{1/2}$; letting $\delta\downarrow0$ along a sequence completes the proof.
\end{proof}

\begin{proof}[Proof of Theorem~\ref{thm:consistency}(a)]
By Lemmas~\ref{lem:margins} and~\ref{lem:transfer} it suffices to show that $M_n(\eta)=n^{-1}\sum_i\ell_\eta(U_i)$ satisfies $\sup_\eta|M_n(\eta)-M(\eta)|\to0$ almost surely, a uniform law of the type in \citet{Wald1949}. For $\eta\in H$ and $r>0$ let $\omega_{\eta,r}(u)=\sup\{|\ell_\zeta(u)-\ell_\eta(u)|:\operatorname{dist}(\zeta,\eta)<r\}$. By continuity $\omega_{\eta,r}\downarrow0$ pointwise as $r\downarrow0$, and $\omega_{\eta,r}\le2E$, which is integrable; hence $\E\,\omega_{\eta,r}(U)\to0$, and dominated convergence also shows that $M$ is finite and continuous. Given $\varepsilon>0$, choose radii $r_\eta$ with $\E\,\omega_{\eta,r_\eta}(U)<\varepsilon$ and a finite subcover of $H$ by balls $B(\eta_m,r_{\eta_m})$, $m=1,\ldots,N$. For $\zeta\in B(\eta_m,r_{\eta_m})$,
\[
|M_n(\zeta)-M(\zeta)|\le|M_n(\eta_m)-M(\eta_m)|+\frac1n\sum_{i=1}^n\omega_{\eta_m,r_{\eta_m}}(U_i)+\E\,\omega_{\eta_m,r_{\eta_m}}(U),
\]
so by the strong law $\limsup_n\sup_\zeta|M_n(\zeta)-M(\zeta)|\le2\varepsilon$ almost surely.
\end{proof}

\subsection{Proof of Theorem~\ref{thm:consistency}(b) and (c)}\label{app:concentration}

By the envelope every $\ell_\eta$ is $c_0$-integrable, so $M(\eta_0)-M(\eta)=\KL(c_0,q_\eta)\ge\frac12\|c_0-q_\eta\|_1^2$ by Pinsker's inequality. The map $\eta\mapsto q_\eta$ is continuous into $L^1$ by Scheff\'e's lemma, so $D=\{\eta:\|q_\eta-c_0\|_1>\varepsilon\}$ is open, and $M(\eta)<M(\eta_0)-\gamma$ on $D$ with $\gamma=\varepsilon^2/2$. Let $e_n=\sup_\eta|\hM(\eta)-M(\eta)|$, which tends to zero almost surely by part (a). For $0<a<\gamma$ the set $V_a=\{\eta:M(\eta)>M(\eta_0)-a\}$ is an open neighborhood of $\eta_0$, so $\Pi(V_a)>0$ by Assumption~\ref{ass:support}. Bounding $\exp\{n\hM\}$ above on $D$ and below on $V_a$,
\begin{align*}
\hPi(D)=\frac{\int_D\exp\{n\hM(\eta)\}\,\Pi(d\eta)}{\int_H\exp\{n\hM(\eta)\}\,\Pi(d\eta)}
&\le\frac{\exp[n\{M(\eta_0)-\gamma+e_n\}]}{\Pi(V_a)\exp[n\{M(\eta_0)-a-e_n\}]}\\
&=\frac{\exp\{-n(\gamma-a-2e_n)\}}{\Pi(V_a)}\to0 .
\end{align*}
The normalizing integral is finite and positive because $\hM$ is continuous on the compact set $H$. This proves (b). For (c), Fubini's theorem and $\|q_\eta-c_0\|_1\le2$ give $\|\overline q_n-c_0\|_1\le\int\|q_\eta-c_0\|_1\,\hPi(d\eta)\le\varepsilon+2\hPi(D)$, and $\varepsilon$ is arbitrary.

\subsection{Standardization within observed strata}\label{app:strata}

\begin{sproposition}\label{prop:s-strata}
Let $(S_i,X_i)$ be independent copies of $(S,X)$, where $S$ takes values in a finite set $\mathcal S$ with $\Pr(S=s)>0$ and the conditional distribution functions $F_{j\mid s}$ of $X_j$ given $S=s$ are continuous. Let $U=(F_{1\mid S}(X_1),\ldots,F_{d\mid S}(X_d))$ have density $c_0$, and compute $\hU_{ij}$ from ranks, or from kernel estimates with bandwidths tending to zero, within the stratum of observation $i$. Under Assumptions~\ref{ass:envelope} and~\ref{ass:support}, the conclusions of Theorem~\ref{thm:consistency} hold. The target $c_0=\sum_s\Pr(S=s)c_{0\mid s}$ is the stratum-weighted average of the conditional copula densities $c_{0\mid s}$ of $X$ given $S=s$, which need not coincide.
\end{sproposition}

\begin{proof}
Let $p_s=\Pr(S=s)$. Given $S=s$, each $F_{j\mid s}(X_j)$ is uniform, so $\Pr(U_j\le t)=\sum_sp_st=t$: the vectors $U_i$ are independent with uniform margins and density $c_0=\sum_sp_sc_{0\mid s}$. For each $s$, $n_s/n\to p_s>0$ almost surely, and the observations with $S_i=s$, taken in their original order, form an independent sequence with the conditional law of $X$ given $S=s$. The proof of Lemma~\ref{lem:margins} therefore applies within each stratum, with the conditional distribution functions in place of $F_j$; since $\mathcal S$ is finite, $\max_{i,j}|\hU_{ij}-U_{ij}|\to0$. The moment bounds within strata do not depend on $n_s$, and the pooled average in~\eqref{eq:transfer-input} is their average with weights $n_s/n$. Hence~\eqref{eq:transfer-input} holds, and Lemma~\ref{lem:transfer} and Sections~\ref{app:transfer}--\ref{app:concentration} apply without change to the independent vectors $U_i$ with density $c_0$.
\end{proof}

\subsection{Proof of Corollary~\ref{cor:families}}\label{app:families}

Constants below depend only on $d$ and the compact parameter sets.

\emph{Mixtures.} If $e^{-E(u)}\le c_\theta(u)\le e^{E(u)}$ for every component, the same bounds hold for every convex combination $\sum_kw_kc_{\theta_k}$, with or without zero weights; positivity and continuity of the components give joint continuity of $(w,\theta_1,\ldots,\theta_K,u)\mapsto\log\sum_kw_kc_{\theta_k}(u)$. A finite union of compact sets is compact, and envelopes for different families combine by taking the larger constants. It therefore suffices to establish continuity and the two-sided envelope for each family.

\emph{Gaussian.} With $z_j=\Phi^{-1}(u_j)$, $\log c_R(u)=-\frac12\log|R|-\frac12z^\top(R^{-1}-I)z$. For $R\in\Rset$ the eigenvalues of $R$ lie in $[\kappa,d]$, so $|\log c_R(u)|\le C(1+\|z\|^2)$. The tail bound $1-\Phi(x)\le e^{-x^2/2}$ for $x\ge0$ gives $\Phi^{-1}(t)^2\le-2\log\min(t,1-t)\le2b(t)$.

\emph{Student.} With $x_j=t_\nu^{-1}(u_j)$ and $Q=x^\top R^{-1}x$,
\[
\log c_{\nu,R}(u)=a_{\nu,R}-\frac{\nu+d}2\log\left(1+\frac Q\nu\right)+\frac{\nu+1}2\sum_{j=1}^d\log\left(1+\frac{x_j^2}\nu\right),
\]
where $a_{\nu,R}$ is continuous, hence bounded, on $[\nu_1,\nu_2]\times\Rset$. Since $Q\le\|x\|^2/\kappa$ and $1+\|x\|^2\le\prod_j(1+x_j^2)$, we get $|\log c_{\nu,R}(u)|\le C\{1+\sum_j\log(1+x_j^2)\}$. Let $p=\nu_1/2$ and let $m_p=\sup_{\nu\in[\nu_1,\nu_2]}\E|T_\nu|^p$, which is finite for $T_\nu\sim t_\nu$ because $p<\nu$. Markov's inequality gives $\min(t,1-t)=\Pr\{|T_\nu|\ge|t_\nu^{-1}(t)|\}/2\le m_p|t_\nu^{-1}(t)|^{-p}/2$, hence $\log\{1+t_\nu^{-1}(t)^2\}\le C+(2/p)\,b(t)$. Joint continuity follows from that of $(\nu,t)\mapsto t_\nu^{-1}(t)$. Different components may have different $\nu$.

\emph{Clayton.} For $\theta>0$ and $L_j=-\log u_j$,
\[
c_\theta(u)=\prod_{m=1}^{d-1}(1+m\theta)\prod_{j=1}^du_j^{-1-\theta}S^{-d-1/\theta},\qquad S=\sum_{j=1}^du_j^{-\theta}-d+1 .
\]
Induction on $d$ with $(\alpha-1)(\beta-1)\ge0$ for $\alpha,\beta\ge1$ gives $1\le S\le\prod_ju_j^{-\theta}$, hence
\[
\log\prod_{m=1}^{d-1}(1+m\theta)-(d-1)\theta\sum_jL_j\le\log c_\theta(u)\le\log\prod_{m=1}^{d-1}(1+m\theta)+(1+\theta)\sum_jL_j,
\]
and $L_j\le b(u_j)$. As $\theta\downarrow0$, $\log S=\theta\sum_jL_j+O(\theta^2)$ uniformly on compact subsets of $(0,1)^d$, so $\log c_\theta\to0$ there; thus $(\theta,u)\mapsto\log c_\theta(u)$ is continuous on $[0,M_{\mathrm C}]\times(0,1)^d$ with the independence copula at $\theta=0$.

\emph{Gumbel.} For $\gamma\ge1$ write $a=1/\gamma$, $x_j=-\log u_j$ and $S=\sum_jx_j^\gamma$. The generator $\psi(s)=\exp(-s^a)$ satisfies $(-1)^d\psi^{(d)}(s)=e^{-s^a}H_d(s)$ with $H_d(s)=\sum_{k=1}^dA_{d,k}s^{ak-d}$, where $A_{1,1}=a$ and $A_{d+1,k}=(d-ak)A_{d,k}+aA_{d,k-1}$ (with $A_{d,0}=A_{d,d+1}=0$). The coefficients are nonnegative because $a\le1$, continuous in $a$, and $A_{d,d}=a^d$. As $ak-d\in[-d,0]$, this gives $|\log H_d(S)|\le C+d|\log S|$. The density is
\[
\log c_\gamma(u)=\sum_jx_j-S^a+d\log\gamma+(\gamma-1)\sum_j\log x_j+\log H_d(S).
\]
Here $0\le\sum_jx_j-S^a\le\sum_jx_j$, because $S^a$ is the $\ell^\gamma$ norm of $(x_1,\ldots,x_d)$; $|\log S|\le\log d+\gamma\sum_j|\log x_j|$; $x_j\le b(u_j)$; and $|\log x_j|\le b(u_j)$, because $\log x_j<x_j$ and $x_j\ge1-u_j$. This proves the envelope. At $\gamma=1$, $H_d\equiv1$ and $c_1\equiv1$, and all expressions are continuous in $(\gamma,u)$.

\emph{Frank.} For $\theta>0$ let $q=1-e^{-\theta}$ and $y=\prod_j(1-e^{-\theta u_j})/q^{d-1}\in[0,q]$. Differentiating the generator $\psi(s)=-\theta^{-1}\log(1-qe^{-s})$ gives
\[
c_\theta(u)=\left(\frac\theta q\right)^{d-1}e^{-\theta\sum_ju_j}\sum_{m=1}^\infty m^{d-1}y^{m-1}.
\]
Since $0\le y\le q\le1-e^{-M_{\mathrm F}}<1$, $c_\theta$ is continuous and bounded above and away from zero on $[0,M_{\mathrm F}]\times[0,1]^d$, with $c_0\equiv1$. For $d=2$ and $\theta\ne0$,
\[
c_\theta(u,v)=\frac{\theta(1-e^{-\theta})e^{-\theta(u+v)}}{D_\theta(u,v)^2},
\]
where $D_\theta(u,v)=e^{-\theta u}\bigl(1-e^{-\theta(1-u)}\bigr)+e^{-\theta v}\bigl(1-e^{-\theta u}\bigr)$.
Both terms of $D_\theta$ have the sign of $\theta$ and they do not vanish simultaneously on $[0,1]^2$, and $D_\theta/\theta\to1$ as $\theta\to0$; hence the density extends continuously and positively to any compact $\theta$-interval times $[0,1]^2$.

\emph{Prior support.} A Dirichlet distribution with positive parameters has a positive density on the interior of $\Delta_{K-1}$, and every relatively open subset of $\Delta_{K-1}$, including one that meets the boundary, contains an open subset of the interior. Hence the product prior charges every neighborhood of every point of $\Delta_{K-1}\times\Theta^K$.

\subsection{Proof of Proposition~\ref{prop:functionals}}\label{app:functionals}

Let $\varepsilon>0$. We claim there is $\delta>0$ such that $\operatorname{dist}(g(\eta),g(\eta_0))\le\varepsilon$ whenever $\|q_\eta-c_0\|_1<\delta$. Otherwise there are $\eta_m$ with $\|q_{\eta_m}-c_0\|_1\to0$ and $\operatorname{dist}(g(\eta_m),g(\eta_0))>\varepsilon$. By compactness a subsequence converges to some $\eta^*$; continuity of $\eta\mapsto q_\eta$ in $L^1$ gives $q_{\eta^*}=c_0=q_{\eta_0}$, so $g(\eta^*)=g(\eta_0)$, contradicting the continuity of $g$. Hence $\hPi\{\operatorname{dist}(g(\eta),g(\eta_0))>\varepsilon\}\le\hPi\{\|q_\eta-c_0\|_1\ge\delta\}\to0$ by Theorem~\ref{thm:consistency}.

\emph{Quantities listed after Proposition~\ref{prop:functionals}.} For the families of Corollary~\ref{cor:families} each of these is a continuous function of $\eta$ that depends on $\eta$ only through $q_\eta$; since the densities are continuous, $q_\eta=q_{\eta'}$ almost everywhere implies equality everywhere.

(i) \emph{Functionals of the distribution function.} Let $C_\eta$ and $c_\eta$ be the distribution function and density of a bivariate margin of $q_\eta$. Then $\sup_{u,v}|C_\eta(u,v)-C_{\eta'}(u,v)|\le\|q_\eta-q_{\eta'}\|_1$ and $\|c_\eta-c_{\eta'}\|_1\le\|q_\eta-q_{\eta'}\|_1$, and $\eta\mapsto q_\eta$ is continuous in $L^1$. Hence Spearman's $\rho=12\int C_\eta-3$ and $\chi_L(t)$, $\chi_U(t)$ at fixed $t$ are continuous in $\eta$, and so is Kendall's $\tau=4\int C_\eta\,c_\eta-1$, because $|\int C_mc_m-\int Cc|\le\|C_m-C\|_\infty+\|c_m-c\|_1$.

(ii) \emph{Tail coefficients.} For a bivariate mixture,
\[
\lambda_L=\lim_{t\downarrow0}\sum_kw_k\frac{C_{\theta_k}(t,t)}t=\sum_kw_k\lambda_L(c_{\theta_k}),
\]
and the component coefficients ($2^{-1/\theta}$ for Clayton, extended by $0$ at $\theta=0$, and $0$ for Frank and Gaussian) are continuous on the sets of Corollary~\ref{cor:families}; the same holds for $\lambda_U$, with $2-2^{1/\gamma}$ for Gumbel, and for the Student coefficient $2t_{\nu+1}[-\{(\nu+1)(1-\rho)/(1+\rho)\}^{1/2}]$. This continuity in the parameter is what the argument uses; compactness and $L^1$ consistency of the density do not by themselves imply consistency of tail coefficients.

(iii) \emph{Conditional probabilities at a point.} Fix $a\in(0,1)^{d-1}$ and $t\in(0,1)$. The $(d-1)$-dimensional margin of each component belongs to the same family: an exchangeable Clayton, Gumbel or Frank copula with the same parameter, a Gaussian or Student copula with the principal submatrix $R_{-d}$, whose smallest eigenvalue is at least $\kappa$ by Cauchy interlacing, or the uniform density if $d=2$. Its density $m_\theta(a)$ is positive and jointly continuous in $(\theta,a)$, by the arguments of Section~\ref{app:families} in dimension $d-1$. For these families the conditional distribution function $v\mapsto\Pr_\theta(U_d\le v\mid U_{-d}=a)$ is available in closed form, is continuous on $[0,1]$ and equals $0$ and $1$ at the endpoints, so $\int_0^1c_\theta(a,v)\,dv=m_\theta(a)$ for every $a$. For the mixture let $m_\eta(a)=\sum_kw_km_{\theta_k}(a)>0$ and $f_\eta(v)=q_\eta(a,v)/m_\eta(a)$, a probability density on $(0,1)$ that is determined by $q_\eta$. If $\eta_m\to\eta$ in $H$, then $f_{\eta_m}(v)\to f_\eta(v)$ for every $v$, by continuity of $q_\eta(a,v)$ and $m_\eta(a)$ in $\eta$, and Scheff\'e's lemma gives $\int|f_{\eta_m}-f_\eta|\to0$. Hence $\Pr_\eta(U_d>t\mid U_{-d}=a)=\int_t^1f_\eta$ is continuous in $\eta$. This route through the model is needed: $L^1$ consistency of the joint density alone does not control conditional probabilities on the null set $\{U_{-d}=a\}$.

\subsection{Proof of Theorem~\ref{thm:identification}}\label{app:identification}

\emph{Gaussian copulas.} Suppose $\sum_{k=1}^m\beta_kc_{R_k}=0$ on $(0,1)^d$ for distinct nonsingular $R_1,\ldots,R_m$. Multiplying by $\prod_j\phi(z_j)$ with $z_j=\Phi^{-1}(u_j)$ gives $\sum_k\beta_k\phi_{R_k}(z)=0$ for all $z\in\R^d$, where $\phi_R$ is the $N(0,R)$ density. Each set $\{v:v^\top(R_k^{-1}-R_l^{-1})v=0\}$, $k\ne l$, is the zero set of a nonzero quadratic form and has Lebesgue measure zero, so there is $v$ for which $a_k=v^\top R_k^{-1}v$ are pairwise distinct. Along $z=tv$ and with $s=t^2$, $\sum_k\beta_k|R_k|^{-1/2}e^{-a_ks/2}=0$ for all $s\ge0$. Differentiating $m-1$ times at $s=0$ gives a Vandermonde system in $a_1,\ldots,a_m$, so every $\beta_k=0$. The same argument, with $(1+s/\nu)^{-(\nu+d)/2}$ in place of $e^{-s/2}$, applies to Student copulas with a common $\nu$, because all derivatives of this function at $s=0$ are nonzero.

\emph{Dimension two.} Suppose
\begin{equation}\label{eq:relation}
\sum_k\beta_kc_k(u,v)=0,\qquad(u,v)\in(0,1)^2,
\end{equation}
for finitely many distinct members $c_k$ of the family; by continuity, a relation that holds almost everywhere holds everywhere. Put $L=-\log u$ and $y=-\log v$. As $L\to\infty$ with $v$ fixed,
\begin{align}
c^{\mathrm G}_\gamma(u,v)&=\{1+o(1)\}\,y^{\gamma-1}v^{-1}L^{-(\gamma-1)}, &&\gamma>1,\label{eq:gumbel-edge}\\
c^{\mathrm C}_\theta(u,v)&=\{1+o(1)\}\,(1+\theta)v^{-1-\theta}e^{-\theta L}, &&\theta>0,\label{eq:clayton-edge}\\
c^{\mathrm N}_\rho(u,v)&=\exp\{-\rho^2L/(1-\rho^2)+O(L^{1/2})\}, &&\rho\ne0.\label{eq:gauss-edge}
\end{align}
For~\eqref{eq:gumbel-edge}, the bivariate Gumbel density equals $e^{L+y-S^{1/\gamma}}(Ly)^{\gamma-1}S^{2/\gamma-2}\{1+(\gamma-1)S^{-1/\gamma}\}$ with $S=L^\gamma+y^\gamma$, and $S^{1/\gamma}=L+\gamma^{-1}y^\gamma L^{1-\gamma}+o(L^{1-\gamma})$. Equation~\eqref{eq:clayton-edge} follows by factoring $u^{-\theta}$ out of $u^{-\theta}+v^{-\theta}-1$ in the Clayton density, and~\eqref{eq:gauss-edge} from $\Phi^{-1}(u)^2=2L+O(\log L)$.

\emph{Step 1.} As $u\downarrow0$, a Frank density converges to $\theta(1-e^{-\theta})^{-1}e^{-\theta v}$, the independence density equals one, and the Clayton, Gumbel and nonindependent Gaussian densities tend to zero by~\eqref{eq:gumbel-edge}--\eqref{eq:gauss-edge}. The limit of~\eqref{eq:relation} is therefore a linear relation among distinct exponentials $e^{-\theta v}$, with $\theta=0$ for independence, on the interval $(0,1)$; distinct exponentials are linearly independent, so all Frank and independence coefficients vanish.

\emph{Step 2.} By~\eqref{eq:clayton-edge} and~\eqref{eq:gauss-edge}, Clayton and Gaussian terms are $o(L^{-r})$ for every $r>0$. Dividing~\eqref{eq:relation} by the Gumbel term with the smallest $\gamma$ and letting $L\to\infty$ shows that its coefficient vanishes; repeating the argument removes every Gumbel term.

\emph{Step 3.} Take $v=\Phi(1)$ and $u=\Phi(-x)$ with $x\to\infty$. Then, exactly,
\[
c^{\mathrm N}_\rho\{\Phi(-x),\Phi(1)\}=A_\rho\exp\left\{-\frac{\rho^2x^2}{2(1-\rho^2)}-\frac{\rho x}{1-\rho^2}\right\},\qquad A_\rho>0,
\]
and by~\eqref{eq:clayton-edge} and Mills' ratio, $c^{\mathrm C}_\theta\{\Phi(-x),\Phi(1)\}=\{1+o(1)\}B_\theta x^{-\theta}e^{-\theta x^2/2}$ with $B_\theta>0$. The exponent pairs $\bigl(\rho^2/[2(1-\rho^2)],\ \rho/(1-\rho^2)\bigr)$ of the Gaussian terms and $(\theta/2,0)$ of the Clayton terms are pairwise distinct, because $\rho\mapsto\rho/(1-\rho^2)$ is strictly increasing and nonzero for $\rho\ne0$. Order the pairs lexicographically. The term with the smallest pair exceeds every other term by a factor $\exp(cx^2)$ or $\exp(cx)$ with $c>0$, which no power of $x$ can offset; dividing~\eqref{eq:relation} by it and letting $x\to\infty$ shows that its coefficient vanishes, and induction removes the rest. This proves the theorem for $d=2$, with Frank parameters of either sign.

\emph{Dimension $d>2$.} Integrate a relation $\sum_k\beta_kc_k=0$ on $(0,1)^d$ over $u_3,\ldots,u_d$. The bivariate margin of an exchangeable Clayton, Gumbel or Frank copula belongs to the same family with the same parameter, that of a Gaussian copula with matrix $R$ is the bivariate Gaussian copula with correlation $R_{12}$, and that of independence is independence. After grouping Gaussian and independence terms with equal margins, each Archimedean component appears in the bivariate relation with its own coefficient, which vanishes by the case $d=2$. The original relation then involves only distinct Gaussian copulas, independence being $R=I$, and their coefficients vanish by the first part of the proof.

\subsection{Proof of Corollary~\ref{cor:mixing}}\label{app:mixing}

The set $\mathcal C=\{c_\theta:\theta\in\Theta\}$ with the $L^1$ metric is compact, as the continuous image of $\Theta$. For fixed $u\in(0,1)^d$, the evaluation $c\mapsto c(u)$ is well defined and continuous on $\mathcal C$: the map $\theta\mapsto c_\theta$ is a closed surjection from a compact space onto a Hausdorff space, hence a quotient map, and $\theta\mapsto c_\theta(u)$ is continuous and constant on its fibres. The set $\mathcal P_K(\mathcal C)$ of probability measures on $\mathcal C$ with at most $K$ atoms is compact under $W$, as the continuous image of $\Delta_{K-1}\times\mathcal C^K$. If $G_m\to G$ in $W$, then $q_{G_m}(u)=\int c(u)\,G_m(dc)\to q_G(u)$ for every $u$, since $c\mapsto c(u)$ is bounded and continuous on $\mathcal C$, and Scheff\'e's lemma gives $\|q_{G_m}-q_G\|_1\to0$. By Theorem~\ref{thm:identification}, or by the classical results for Gaussian and common-$\nu$ Student copulas, the map $G\mapsto q_G$ is injective on $\mathcal P_K(\mathcal C)$, hence a homeomorphism onto its image. Consequently $g(\eta)=G_\eta$ is continuous and depends on $\eta$ only through $q_\eta$, and Proposition~\ref{prop:functionals} gives (a).

For (b), let $G_0=G_{\eta_0}=\sum_{k=1}^Kw_{0k}\delta_{c_{0k}}$ with distinct $c_{0k}$ and $w_{0k}>0$, fix $r$ smaller than half the smallest distance between the $c_{0k}$, and let $N_k$ be the open $L^1$ ball of radius $r$ about $c_{0k}$. The function $f_k(c)=\max\{0,r-\|c-c_{0k}\|_1\}$ is $1$-Lipschitz and vanishes off $N_k$, so by Kantorovich--Rubinstein duality $\int f_k\,dG\ge rw_{0k}-W(G,G_0)$. If $W(G,G_0)<\delta<r\min_kw_{0k}$, each $N_k$ therefore contains an atom of $G$; as $G$ has at most $K$ atoms and the $N_k$ are disjoint, each contains exactly one, say $c_{\sigma(k)}$, with weight $w_{\sigma(k)}=G(N_k)\ge\int f_k\,dG/r\ge w_{0k}-\delta/r$. Since the weights sum to one, also $w_{\sigma(k)}\le w_{0k}+(K-1)\delta/r$. Choosing $r$ and then $\delta$ small makes $\max_k\{|w_{\sigma(k)}-w_{0k}|+\|c_{\sigma(k)}-c_{0k}\|_1\}$ smaller than any given $\varepsilon$, and (b) follows from (a).

For (c), the finitely many atoms of $G_0$ lie outside the closed set $\mathcal C_F$, so
\[
\delta_0=\min_{c\in\operatorname{supp}G_0}\operatorname{dist}(c,\mathcal C_F)>0,
\]
distances being in $L^1$. The function $f(c)=\min\{\operatorname{dist}(c,\operatorname{supp}G_0),\delta_0\}$ is $1$-Lipschitz, vanishes on $\operatorname{supp}G_0$ and equals $\delta_0$ on $\mathcal C_F$. By Kantorovich--Rubinstein duality, $\delta_0G(\mathcal C_F)\le\int f\,dG-\int f\,dG_0\le W(G,G_0)$. Hence $\{\eta:G_\eta(\mathcal C_F)>\varepsilon\}\subset\{\eta:W(G_\eta,G_0)>\delta_0\varepsilon\}$, and (c) follows from (a). Since $\eta\mapsto G_\eta(\mathcal C_F)=\sum_kw_k\mathbf 1\{c_{\theta_k}\in\mathcal C_F\}$ is Borel, the event in (c) is measurable.

\emph{Class probabilities.} The next result is mentioned after Corollary~\ref{cor:mixing}.

\begin{sproposition}\label{prop:s-class}
Under the conditions of Corollary~\ref{cor:mixing}(b), the class probabilities $r_{\eta,k}(u)=w_kc_{\theta_k}(u)/q_\eta(u)$ satisfy $\hPi\{\eta:\min_\sigma\int\sum_k|r_{\eta,\sigma(k)}-r_{\eta_0,k}|\,c_0>\varepsilon\}\to0$ almost surely for every $\varepsilon>0$. If $K\ge2$ and every $w_{0k}>0$, then $\max_kr_{\eta_0,k}(u)<1$ for every $u$. If $\Pr(Z=k)=w_k$ and $U$ given $Z=k$ has copula density $c_{\theta_k}$, then $Z$ is independent of each coordinate $U_j$.
\end{sproposition}

\begin{proof}
Fix $u\in(0,1)^d$. The map $(w,c_1,\ldots,c_K)\mapsto w_kc_k(u)/\sum_lw_lc_l(u)$ is continuous on $\Delta_{K-1}\times\mathcal C^K$ because the densities are positive, and class probabilities are bounded by one, so by dominated convergence the matched distance $\int\sum_k|r_{\eta,\sigma(k)}-r_{\eta_0,k}|\,c_0$ is continuous at the truth in the sense of Corollary~\ref{cor:mixing}(b); the first claim follows from (b) as in Section~\ref{app:functionals}. If $K\ge2$ and every $w_{0k}>0$, then $r_{\eta_0,k}(u)>0$ for all $k$, so $\max_kr_{\eta_0,k}(u)<1$. Finally, $\Pr(U_j\le t,Z=k)=w_kt=\Pr(U_j\le t)\Pr(Z=k)$ for every $j$ and $t$, so $Z$ is independent of each coordinate $U_j$.
\end{proof}

\subsection{Proof of Corollary~\ref{cor:fourfamily}}\label{app:fourfamily}

Each $\Theta_k$ is a compact subset of the corresponding parameter set in Corollary~\ref{cor:families}. Continuity and the envelope of Assumption~\ref{ass:envelope} therefore hold on the compact set $H$ (Section~\ref{app:families}), and the product prior gives positive mass to every open neighborhood of $\eta_0$ in $H$, so Assumption~\ref{ass:support} holds. Theorem~\ref{thm:consistency} and Proposition~\ref{prop:functionals} apply.

Let $c^{(k)}_\theta$ denote the density of family $k$ and $\mathcal C_k=\{c^{(k)}_\theta:\theta\in\Theta_k\}$. The map $\theta\mapsto c^{(k)}_\theta$ is continuous into $L^1$, by pointwise continuity and Scheff\'e's lemma, so $\mathcal C_k$ is compact. It is also injective: Kendall's $\tau$ of a bivariate margin is strictly monotone in the Clayton, Gumbel and Frank parameters, and the correlation matrix of a Gaussian copula is determined by its bivariate margins. Hence $\theta\mapsto c^{(k)}_\theta$ is a homeomorphism of $\Theta_k$ onto $\mathcal C_k$. Because $a_1>0$, $a_2>1$ and $a_3>0$, every element of $\mathcal C=\bigcup_k\mathcal C_k$ is a member of the family of Theorem~\ref{thm:identification}, whose proof treats members as indexed by family and parameter. An element of $\mathcal C_k$ and an element of $\mathcal C_l$, $l\ne k$, carry different indices; if they were equal as functions, their difference would be a nontrivial vanishing linear combination, contradicting the theorem. The four sets are therefore pairwise disjoint and, being compact,
\[
d_{\min}=\min_{k\ne l}\ \inf\{\|c-c'\|_1:c\in\mathcal C_k,\ c'\in\mathcal C_l\}>0 .
\]
For $\eta\in H$ let $G_\eta=\sum_kw_k\delta_{c^{(k)}_{\theta_k}}$, a probability measure on $\mathcal C$ with $G_\eta(\mathcal C_k)=w_k$. The map $\eta\mapsto G_\eta$ is continuous, and $G\mapsto q_G$ is injective on probability measures on $\mathcal C$ with at most four atoms by Theorem~\ref{thm:identification}. The argument for Corollary~\ref{cor:mixing}(a) in Section~\ref{app:mixing} thus applies on $H$ and gives $\hPi\{\eta:W(G_\eta,G_{\eta_0})>\delta\}\to0$ almost surely for every $\delta>0$.

For (a), the function $f_k(c)=\max\{0,1-\operatorname{dist}(c,\mathcal C_k)/d_{\min}\}$ is $d_{\min}^{-1}$-Lipschitz on $\mathcal C$, equals one on $\mathcal C_k$ and vanishes on $\mathcal C_l$ for $l\ne k$. By Kantorovich--Rubinstein duality,
\[
|w_k-w_{0k}|=\Bigl|\int f_k\,dG_\eta-\int f_k\,dG_{\eta_0}\Bigr|\le W(G_\eta,G_{\eta_0})/d_{\min},
\]
and (a) follows. For (b), let $w_{0k}>0$, $c_{0k}=c^{(k)}_{\theta_{0k}}$ and $0<r<d_{\min}$. The function $h(c)=\max\{0,r-\|c-c_{0k}\|_1\}$ is $1$-Lipschitz and vanishes outside the open ball of radius $r$ about $c_{0k}$, which meets no $\mathcal C_l$ with $l\ne k$. Hence $\int h\,dG_{\eta_0}=rw_{0k}$ and $\int h\,dG_\eta=w_kh(c^{(k)}_{\theta_k})\le h(c^{(k)}_{\theta_k})$, so $h(c^{(k)}_{\theta_k})\ge rw_{0k}-W(G_\eta,G_{\eta_0})$. If $W(G_\eta,G_{\eta_0})<rw_{0k}$, then $\|c^{(k)}_{\theta_k}-c_{0k}\|_1<r$. Since $r$ is arbitrary and $\theta\mapsto c^{(k)}_\theta$ is a homeomorphism onto $\mathcal C_k$, (b) follows.

\section{Computation}\label{sec:s-computation}

\subsection{Priors and parameter sets}
All mixture weights have a Dirichlet prior, with parameters $(1,\ldots,1)$ in the simulations and in all fits fitted with Algorithm~1, and $(1/K,\ldots,1/K)$ in the Gaussian and Student fits of Sections~\ref{sec:s-fitness3} and~\ref{sec:s-running}. Component parameters have independent uniform priors on compact sets. In Experiment~1 these are $(0.15,5)$ for the Clayton parameter, $(1.05,4)$ for the Gumbel parameter and $(-0.15,0.90)$ for the exchangeable correlations, the lower limit being raised where needed to keep the matrix positive definite. In Experiments~2 and~3 and in the four-family applications they are $[0.01,12]$ (Clayton), $[1.001,7]$ (Gumbel), $[0.01,30]$ (Frank) and $[-0.95,0.95]$ (Gaussian). Full correlation matrices have a uniform prior on the free off-diagonal entries subject to $\lambda_{\min}(R)\ge0.05$. By Corollary~1 all of these priors satisfy Assumption~2 at every interior parameter.

\subsection{Data-augmentation sampler}
Experiment~1, Experiment~3 (for comparison) and the Gaussian and Student fits of Sections~\ref{sec:s-fitness3} and~\ref{sec:s-running} use the latent-allocation sampler for mixtures \citep{Diebolt1994}. Each iteration draws the allocations from their full conditional, draws the weights from the conjugate Dirichlet conditional, and updates each component parameter by a random-walk Metropolis step that uses only the observations allocated to that component; for bounded scalar parameters the proposal is reflected at the boundary, so that it remains symmetric. For full correlation matrices the proposal perturbs all free entries of one component and is rejected if the smallest eigenvalue falls below 0.05. In Experiment~1 an additional Metropolis block move updates the weight and both component parameters jointly on the logistic scale, using the marginal pseudo-likelihood without allocations; its proposal covariance is estimated during burn-in \citep{Haario2001}. Proposal scales are adapted towards acceptance rates between 0.25 and 0.35 during burn-in only, so the retained draws come from a time-homogeneous chain \citep{Roberts2009}. Experiment~1 uses 5000 iterations with 1500 burn-in iterations and thinning by five; the 61 of 810 fits (7.5\%) in which an effective sample size fell below 100 were rerun with 20000 iterations. The two-component fits of Section~\ref{sec:s-fitness3} use Metropolis updates on the marginal pseudo-posterior, because their allocation chains mixed poorly. Draws of same-family mixtures are relabeled by increasing correlation, which does not affect any reported functional.

\subsection{Marginal sampler}
Algorithm~1 of the main article is used for Experiments~2 and~3 and for all four-family and Gaussian fits to the student and adult data in Sections~6.2--6.4. Each iteration performs three random-walk updates of the log-ratio weights $y$, one random-walk update of each logit parameter $z_k$, independence proposals of each component parameter from its prior every $m=10$ iterations, and one block random-walk update of $(y,z)\in\R^{2K-1}$. The $n\times K$ matrix of component log densities at the pseudo-observations is cached, so weight updates require no density evaluation and an update of one component parameter requires one new column; an iteration thus evaluates about $2K+K/m$ columns, against $K$ for data augmentation. During burn-in, the scales of the single-coordinate updates are adapted every 100 iterations towards an acceptance rate of 0.3, and from iteration 400 the block proposal has covariance $(2.38^2/(2K-1))\widehat\Sigma+10^{-6}I$, where $\widehat\Sigma$ is the covariance of the burn-in history \citep{Haario2001}. All adaptation stops after burn-in. Experiments~2 and~3 use one chain of 8000 iterations with 2000 burn-in iterations and thinning by six. In the applications, the training fits use four chains of 10000 iterations (3000 burn-in, thinning by seven), the full-sample four-family fits four chains of 12000 iterations (3000 burn-in, thinning by nine), and the full-sample Gaussian fits and all sensitivity fits (random ties, sex, year of study, age group) two chains of 6000 to 12000 iterations; every chain retains 1000 draws.

\subsection{Maximum pseudo-likelihood}
The MPL competitor maximizes $\sum_i\log q_\eta(\hU_i)$ over the same parameter set as the Bayesian fit. Weights are parameterized by a logit ($K=2$) or by additive log-ratios ($K=4$) restricted to $[-10,10]$, and component parameters are restricted to their prior intervals. The objective is minimized by L-BFGS-B from ten random starting values, with weights drawn from a Dirichlet$(2,\ldots,2)$ distribution and component parameters drawn uniformly on the scale of Kendall's $\tau$. The best of the ten solutions is reported. A fit is counted as \emph{on the boundary} if a weight is below $10^{-3}$, a log-ratio reaches $\pm10$, or a parameter lies within $10^{-4}$ of its interval end, relative to the interval width; it is counted as \emph{not converged} if the optimizer reports failure for the best start. As a measure of instability we record the largest difference in component Kendall's $\tau$ among starts whose log pseudo-likelihood is within one unit of the best.

\subsection{Accuracy measures}
The divergence $\KL(c_0,\widehat c\,)=\int c_0\log(c_0/\widehat c\,)$ is estimated by averaging $\log\{c_0(V_l)/\widehat c(V_l)\}$ over $16384$ draws $V_l$ from the true copula, the same draws being used for every replicate and method. For the Bayesian fits $\widehat c$ is the posterior predictive density, averaged over 64 evenly spaced posterior draws; for MPL it is the plug-in density. The posterior mean is the point estimate for weights, component parameters, Kendall's $\tau$ and tail coefficients.

\section{Additional simulation results}\label{sec:s-simulation}

Tables~\ref{tab:s-exp1-kl}--\ref{tab:s-exp1-cov} give the full results of Experiment~1. At $n=200$ MPL is less accurate than the pseudo-posterior for every model and every choice of margins. The difference is not caused by the single boundary solution per model, whose removal changes the MPL errors little; it reflects the flatness of the pseudo-likelihood in small samples, over which the posterior averages, and, for the weight, some shrinkage of the posterior mean towards the prior mean $0.5$, close to the true $0.55$. The two methods coincide for $n\ge800$. The oracle shows the cost of estimating the margins, which is small already at $n=800$. Kernel margins are a little less accurate than ranks; Table~\ref{tab:s-cdf} shows that their first-stage error decreases more slowly for N5 and t5, where most of the kernel error comes from the cubed-normal coordinate, whose density is unbounded at zero. Credible intervals cover the true parameter in 80\% to 100\% of the replicates; with 30 replicates these frequencies are imprecise, and nominal coverage is not expected because the pseudo-posterior ignores the estimation of the margins.

\begin{table}[htbp]
\centering\footnotesize
\setlength{\tabcolsep}{3.5pt}
\setstretch{1.1}
\caption{Experiment 1: predictive Kullback--Leibler divergence ($\times10^{3}$), mean (standard error) over 30 replicates.}
\label{tab:s-exp1-kl}
\begin{adjustbox}{max width=\textwidth}
\begin{tabular}{@{}llcccccc@{}}
\toprule
Model & $n$ & \multicolumn{3}{c}{Bayes: oracle, rank, kernel} & \multicolumn{3}{c}{MPL: oracle, rank, kernel}\\
\midrule
CG2 & 200 & 7.29 (0.87) & 7.78 (1.04) & 8.52 (1.12) & 9.12 (1.45) & 9.05 (1.40) & 9.88 (1.55)\\
CG2 & 800 & 2.27 (0.34) & 2.83 (0.41) & 3.62 (0.51) & 2.22 (0.33) & 2.76 (0.39) & 3.73 (0.56)\\
CG2 & 3200 & 0.56 (0.09) & 0.67 (0.11) & 0.89 (0.17) & 0.54 (0.09) & 0.65 (0.10) & 0.91 (0.17)\\
N5 & 200 & 7.43 (1.01) & 10.05 (1.85) & 13.93 (2.80) & 10.64 (1.81) & 14.26 (2.55) & 20.11 (4.09)\\
N5 & 800 & 2.12 (0.34) & 2.46 (0.39) & 3.92 (0.44) & 2.33 (0.39) & 2.77 (0.42) & 4.36 (0.49)\\
N5 & 3200 & 0.49 (0.06) & 0.50 (0.06) & 0.84 (0.12) & 0.48 (0.06) & 0.52 (0.06) & 0.88 (0.12)\\
t5 & 200 & 6.76 (1.15) & 6.58 (1.06) & 6.72 (0.79) & 8.72 (1.33) & 9.25 (1.20) & 8.81 (1.03)\\
t5 & 800 & 1.31 (0.24) & 1.34 (0.25) & 1.64 (0.32) & 1.47 (0.28) & 1.52 (0.30) & 1.90 (0.37)\\
t5 & 3200 & 0.49 (0.08) & 0.47 (0.08) & 0.64 (0.09) & 0.52 (0.08) & 0.51 (0.08) & 0.70 (0.10)\\
\bottomrule
\end{tabular}
\end{adjustbox}
\end{table}

\begin{table}[htbp]
\centering\footnotesize
\setlength{\tabcolsep}{3.5pt}
\setstretch{1.1}
\caption{Experiment 1: root mean squared error of the weight $w$ and of the two component parameters, for the oracle pseudo-posterior (O), the rank-based pseudo-posterior (B) and rank-based MPL (M).}
\label{tab:s-exp1-rmse}
\begin{adjustbox}{max width=\textwidth}
\begin{tabular}{@{}llccccccccc@{}}
\toprule
& & \multicolumn{3}{c}{$w$} & \multicolumn{3}{c}{First parameter} & \multicolumn{3}{c}{Second parameter}\\ Model & $n$ & O & B & M & O & B & M & O & B & M\\
\midrule
CG2 & 200 & 0.148 & 0.158 & 0.160 & 0.553 & 0.424 & 0.541 & 0.366 & 0.349 & 0.504\\
CG2 & 800 & 0.076 & 0.081 & 0.079 & 0.237 & 0.239 & 0.246 & 0.227 & 0.207 & 0.203\\
CG2 & 3200 & 0.042 & 0.044 & 0.044 & 0.100 & 0.099 & 0.099 & 0.088 & 0.096 & 0.097\\
N5 & 200 & 0.102 & 0.099 & 0.153 & 0.069 & 0.077 & 0.108 & 0.057 & 0.052 & 0.082\\
N5 & 800 & 0.054 & 0.052 & 0.056 & 0.037 & 0.040 & 0.042 & 0.033 & 0.038 & 0.040\\
N5 & 3200 & 0.027 & 0.027 & 0.027 & 0.018 & 0.018 & 0.018 & 0.016 & 0.016 & 0.016\\
t5 & 200 & 0.106 & 0.104 & 0.186 & 0.078 & 0.076 & 0.122 & 0.069 & 0.068 & 0.106\\
t5 & 800 & 0.071 & 0.078 & 0.094 & 0.045 & 0.048 & 0.056 & 0.039 & 0.039 & 0.046\\
t5 & 3200 & 0.050 & 0.050 & 0.052 & 0.030 & 0.030 & 0.031 & 0.028 & 0.026 & 0.027\\
\bottomrule
\end{tabular}
\end{adjustbox}
\par\smallskip\footnotesize First and second parameters are $(\theta,\gamma)$ for CG2 and the smaller and larger correlations for N5 and t5.
\end{table}

\begin{table}[htbp]
\centering\footnotesize
\setlength{\tabcolsep}{3.5pt}
\setstretch{1.1}
\caption{Experiment 1: empirical coverage of 95\% credible intervals of the rank-based pseudo-posterior, mean computing time per fit (seconds) and percentage of MPL fits on the boundary of the parameter set or not converged.}
\label{tab:s-exp1-cov}
\begin{adjustbox}{max width=\textwidth}
\begin{tabular}{@{}llccccccc@{}}
\toprule
Model & $n$ & $w$ & First & Second & Bayes (s) & MPL (s) & Boundary (\%) & Not conv. (\%)\\
\midrule
CG2 & 200 & 0.90 & 0.97 & 0.97 & 1.3 & 0.3 & 3 & 0\\
CG2 & 800 & 0.97 & 0.90 & 0.97 & 1.8 & 0.6 & 0 & 0\\
CG2 & 3200 & 0.80 & 0.97 & 0.93 & 5.0 & 1.3 & 0 & 0\\
N5 & 200 & 0.93 & 0.93 & 0.97 & 1.2 & 0.2 & 3 & 0\\
N5 & 800 & 0.97 & 0.97 & 0.90 & 0.9 & 0.3 & 0 & 0\\
N5 & 3200 & 0.97 & 0.97 & 0.97 & 1.8 & 0.5 & 0 & 0\\
t5 & 200 & 1.00 & 1.00 & 1.00 & 1.8 & 0.3 & 3 & 0\\
t5 & 800 & 0.97 & 1.00 & 0.97 & 1.2 & 0.4 & 0 & 0\\
t5 & 3200 & 0.97 & 0.93 & 0.93 & 2.4 & 0.9 & 0 & 0\\
\bottomrule
\end{tabular}
\end{adjustbox}
\par\smallskip\footnotesize Coverage is descriptive: the pseudo-posterior ignores marginal estimation, so nominal coverage is not guaranteed. Bayes times include the 7.5\% of fits whose chains were extended to 20000 iterations because an effective sample size fell below 100.
\end{table}

\begin{table}[htbp]
\centering\small
\setlength{\tabcolsep}{3.5pt}
\setstretch{1.1}
\caption{Experiment 1: root mean squared error of the estimated uniforms $\widehat U_{ij}$ against the true $U_{ij}$, averaged over coordinates and replicates.}
\label{tab:s-cdf}
\begin{adjustbox}{max width=\textwidth}
\begin{tabular}{@{}llcc@{}}
\toprule
Model & $n$ & Ranks & Kernel\\
\midrule
CG2 & 200 & 0.0286 & 0.0266\\
CG2 & 800 & 0.0156 & 0.0155\\
CG2 & 3200 & 0.0067 & 0.0069\\
N5 & 200 & 0.0301 & 0.0350\\
N5 & 800 & 0.0146 & 0.0227\\
N5 & 3200 & 0.0075 & 0.0170\\
t5 & 200 & 0.0277 & 0.0321\\
t5 & 800 & 0.0147 & 0.0222\\
t5 & 3200 & 0.0071 & 0.0164\\
\bottomrule
\end{tabular}
\end{adjustbox}
\end{table}

Table~\ref{tab:s-highdim} reports two further rank-based experiments with 20 replicates per sample size. In dimension ten the predictive divergence decreases at the same rate as in dimension five. The overfitted experiment fits two components to a single Gaussian copula. The density converges, as Theorem~\ref{thm:consistency} predicts, whereas the parameters need not, because the exact-order condition of Corollary~\ref{cor:mixing}(b) fails: about 80\% of the posterior mass stays outside a neighborhood of any single parameter value and the smaller weight remains near 0.2, so the posterior splits the true component between two nearby components rather than emptying one. This is the behavior that \citet{Rousseau2011} establish for regular overfitted mixtures when the Dirichlet parameter exceeds half the dimension of the component parameter, as it does here ($1>1/2$, the exchangeable correlation being scalar); with a parameter below $1/2$ their theory predicts emptying instead, although it does not cover pseudo-posteriors. Corollary~\ref{cor:mixing}(c) does not apply, because the redundant component belongs to the same family as the true one.

\begin{table}[htbp]
\centering\small
\setlength{\tabcolsep}{3.5pt}
\setstretch{1.1}
\caption{Additional rank-based experiments with 20 replicates: predictive Kullback--Leibler divergence ($\times10^3$, mean and standard error), posterior mass outside a normalized radius 0.2 of the true parameter, and posterior mean of the smaller weight.}
\label{tab:s-highdim}
\begin{adjustbox}{max width=\textwidth}
\begin{tabular}{@{}llccc@{}}
\toprule
Model & $n$ & KL & Mass outside & Smaller weight\\
\midrule
Gaussian, $d=10$ & 200 & 13.67 (2.53) & 0.06 & 0.40\\
Gaussian, $d=10$ & 800 & 2.27 (0.42) & 0.00 & 0.44\\
Gaussian, $d=10$ & 3200 & 0.79 (0.21) & 0.00 & 0.45\\
Student $t_5$, $d=10$ & 200 & 9.01 (1.65) & 0.12 & 0.39\\
Student $t_5$, $d=10$ & 800 & 2.37 (0.40) & 0.00 & 0.43\\
Student $t_5$, $d=10$ & 3200 & 0.70 (0.18) & 0.00 & 0.45\\
Overfitted Gaussian, $d=5$ & 200 & 5.50 (1.37) & 0.80 & 0.23\\
Overfitted Gaussian, $d=5$ & 800 & 1.78 (0.42) & 0.78 & 0.21\\
Overfitted Gaussian, $d=5$ & 3200 & 0.49 (0.07) & 0.79 & 0.20\\
\bottomrule
\end{tabular}
\end{adjustbox}
\par\smallskip\footnotesize Dimension-ten models are equicorrelated two-component mixtures with correlations 0.05 and 0.60. The overfitted model fits $K=2$ to a single Gaussian copula with correlation 0.40; there the density converges while the parameters are not identified, as Theorem~1 and Corollary~2 predict.
\end{table}

Table~\ref{tab:s-exp2} gives further results of Experiment~2. With separated components the pseudo-posterior estimates the component Kendall's $\tau$ more accurately than MPL at every sample size; near independence neither method recovers them well, as expected when the components are almost indistinguishable. Credible intervals for the weights cover the truth in 87\% to 91\% of the separated datasets, below the nominal level, and in nearly all near-independent datasets because they are then very wide. Effective sample sizes fall as $n$ grows near independence, where the posterior spreads along ridges of nearly equivalent mixtures; Experiment~3 in the main article shows that this is a property of the posterior rather than of the sampler.

Table~\ref{tab:s-exp2_equal} repeats the experiment with equal true weights $0.25$. The conclusions are the same, with one caution: near independence the posterior mean of the weights then stays close to the prior mean $(0.25,\ldots,0.25)$, which here equals the truth, so its small weight error reflects the prior rather than the data. The unequal weights of the main design avoid this coincidence.

\begin{table}[htbp]
\centering\scriptsize
\setlength{\tabcolsep}{3.5pt}
\setstretch{1.1}
\caption{Experiment 2 with true weights $(0.4,0.3,0.2,0.1)$: root mean squared error of the four weights (oracle pseudo-posterior O, rank-based pseudo-posterior B, rank-based MPL M), of the four component Kendall's $\tau$ and of the tail coefficients $(\lambda_L,\lambda_U)$; and, for the rank-based pseudo-posterior, the average coverage of 95\% credible intervals for the weights, the median effective sample size of the weights and the mean computing time in seconds, with the MPL time for comparison.}
\label{tab:s-exp2}
\begin{adjustbox}{max width=\textwidth}
\begin{tabular}{@{}llccccccccccc@{}}
\toprule
& & \multicolumn{3}{c}{Weights} & \multicolumn{2}{c}{Kendall's $\tau$} & \multicolumn{2}{c}{$(\lambda_L,\lambda_U)$} & & & \multicolumn{2}{c}{Time (s)}\\ Scenario & $n$ & O & B & M & B & M & B & M & Cov. & ESS & B & M\\
\midrule
Separated & 500 & 0.131 & 0.143 & 0.171 & 0.075 & 0.109 & 0.092 & 0.110 & 0.91 & 138 & 49 & 6\\
 & 2000 & 0.108 & 0.113 & 0.120 & 0.060 & 0.076 & 0.073 & 0.079 & 0.87 & 97 & 102 & 14\\
 & 8000 & 0.055 & 0.061 & 0.065 & 0.048 & 0.068 & 0.034 & 0.035 & 0.87 & 112 & 260 & 42\\
Near independence & 500 & 0.191 & 0.189 & 0.318 & 0.259 & 0.305 & 0.021 & 0.038 & 1.00 & 193 & 39 & 5\\
 & 2000 & 0.196 & 0.194 & 0.272 & 0.204 & 0.235 & 0.018 & 0.030 & 1.00 & 99 & 84 & 14\\
 & 8000 & 0.155 & 0.168 & 0.242 & 0.159 & 0.255 & 0.012 & 0.016 & 0.98 & 64 & 281 & 63\\
\bottomrule
\end{tabular}
\end{adjustbox}
\end{table}

\begin{table}[htbp]
\centering\scriptsize
\setlength{\tabcolsep}{3.5pt}
\setstretch{1.1}
\caption{Experiment 2 with true weights $(0.25,0.25,0.25,0.25)$: root mean squared error of the four weights (oracle pseudo-posterior O, rank-based pseudo-posterior B, rank-based MPL M), of the four component Kendall's $\tau$ and of the tail coefficients $(\lambda_L,\lambda_U)$; and, for the rank-based pseudo-posterior, the average coverage of 95\% credible intervals for the weights, the median effective sample size of the weights and the mean computing time in seconds, with the MPL time for comparison.}
\label{tab:s-exp2_equal}
\begin{adjustbox}{max width=\textwidth}
\begin{tabular}{@{}llccccccccccc@{}}
\toprule
& & \multicolumn{3}{c}{Weights} & \multicolumn{2}{c}{Kendall's $\tau$} & \multicolumn{2}{c}{$(\lambda_L,\lambda_U)$} & & & \multicolumn{2}{c}{Time (s)}\\ Scenario & $n$ & O & B & M & B & M & B & M & Cov. & ESS & B & M\\
\midrule
Separated & 500 & 0.121 & 0.130 & 0.172 & 0.077 & 0.169 & 0.071 & 0.096 & 0.96 & 185 & 44 & 5\\
 & 2000 & 0.113 & 0.123 & 0.125 & 0.056 & 0.081 & 0.070 & 0.072 & 0.84 & 153 & 67 & 9\\
 & 8000 & 0.057 & 0.063 & 0.053 & 0.026 & 0.029 & 0.035 & 0.029 & 0.88 & 193 & 315 & 54\\
Near independence & 500 & 0.106 & 0.107 & 0.302 & 0.268 & 0.306 & 0.024 & 0.041 & 0.99 & 175 & 39 & 5\\
 & 2000 & 0.114 & 0.114 & 0.261 & 0.203 & 0.236 & 0.014 & 0.022 & 1.00 & 115 & 70 & 14\\
 & 8000 & 0.151 & 0.143 & 0.232 & 0.168 & 0.224 & 0.012 & 0.016 & 0.98 & 59 & 179 & 45\\
\bottomrule
\end{tabular}
\end{adjustbox}
\end{table}

\section{Data}\label{sec:s-data}

\subsection{Student fitness tests}
The Jiangxi data are the supporting dataset S1 of \citet{Sun2024}; we used the published workbook without correction. All 8772 records have the three tests. A sensitivity sample of 8648 records excludes 123 students whose recorded age lies outside the 17--22 years stated by the source and one record whose sex-specific strength test is inconsistent with the recorded sex. The training and evaluation split (3000 and 5772 records) is stratified by sex and grade. Evaluation pseudo-observations use the within-sex empirical distribution functions of the training records, with average ranks for ties. The source describes a multistage stratified cluster sample from four universities but does not identify the university of each student, so no clustering is taken into account. Year of study is recorded as first, second or third year.

\subsection{Adults: NHANES 2011--2014}
We downloaded the demographic (DEMO\_G, DEMO\_H), muscle strength (MGX\_G, MGX\_H) and physical activity monitor day summary (PAXDAY\_G, PAXDAY\_H) files of the 2011--2012 and 2013--2014 cycles from the NCHS website \citep{NHANES}; their MD5 checksums are stored with the data. Variables are sex (RIAGENDR, as reported at the interview), age in years (RIDAGEYR, top-coded at 80), the grip test status (MGDEXSTS) and combined grip strength (MGDCGSZ, kg, the sum of the largest reading of each hand), and, for each monitored day, the wake-wear and sleep-wear minutes (PAXWWMD, PAXSWMD) and the daily MIMS total (PAXMTSD).

The two cycles contain 19931 participants, of whom 10907 adults aged 20 to 80 years have a grip record. Of these, 9794 completed the grip test (MGDEXSTS $=1$ and MGDCGSZ $>0$). A monitored day is valid if its wake-wear and sleep-wear minutes add up to at least 1368, that is, 95\% of the 1440 minutes; the partial first and last days of each recording therefore rarely qualify. We keep participants with at least four valid days and define activity as the mean daily MIMS total over their valid days. The analysis sample contains 5336 adults: 999, 1123 and 959 women and 805, 786 and 664 men in the age groups 20--39, 40--59 and 60--80 years. The training sample of 2499 participants is drawn at random within each of these six strata in proportion to its size; the remaining 2837 form the evaluation sample. Evaluation pseudo-observations use the within-stratum empirical distribution functions of the training sample, with the mid-rank convention for ties. The examination weights and design variables (WTMEC2YR, SDMVPSU, SDMVSTRA) are retained in the prepared file but not used.

\section{Additional application results}\label{sec:s-applications}

\subsection{Three-dimensional student fitness models}\label{sec:s-fitness3}
Table~\ref{tab:s-fitness3} reports the two-component Gaussian mixture for speed, power and flexibility. The weights are uncertain and the class probabilities of individual students are diffuse (mean normalized entropy 0.96), so the components describe two coupling regimes for speed and power rather than two groups of students. Table~\ref{tab:s-fitness3-sens} shows that excluding flagged records or breaking ties at random leaves the fit essentially unchanged. In the same model, the probability that a student's flexibility lies in the lowest within-sex quartile ranges only from 0.21 to 0.29 across the nine combinations of speed and power quartiles; it decreases with jump distance and hardly depends on sprint speed. Sprint and jump results therefore carry little information about flexibility.

\begin{table}[htbp]
\centering\small
\setlength{\tabcolsep}{3.5pt}
\setstretch{1.1}
\caption{Two-component Gaussian mixture for the standardized speed, power and flexibility scores of 8772 students: posterior means and 95\% credible intervals of the weights and normal-score correlations. Components are ordered by their speed--power correlation.}
\label{tab:s-fitness3}
\begin{adjustbox}{max width=\textwidth}
\begin{tabular}{@{}lcc@{}}
\toprule
Quantity & Component 1 & Component 2\\
\midrule
Weight & 0.52 (0.20, 0.81) & 0.48 (0.19, 0.80)\\
Speed--power & 0.27 (0.09, 0.37) & 0.63 (0.53, 0.76)\\
Speed--flexibility & 0.07 (-0.01, 0.16) & 0.07 (-0.03, 0.16)\\
Power--flexibility & 0.11 (-0.01, 0.20) & 0.20 (0.10, 0.32)\\
\bottomrule
\end{tabular}
\end{adjustbox}
\end{table}

\begin{table}[htbp]
\centering\footnotesize
\setlength{\tabcolsep}{3.5pt}
\setstretch{1.1}
\caption{Sensitivity of the two-component Gaussian fit: posterior mean weight of component 1, component speed--power correlations, mean normalized entropy of the class probabilities and percentage of students whose largest class probability exceeds 0.8.}
\label{tab:s-fitness3-sens}
\begin{adjustbox}{max width=\textwidth}
\begin{tabular}{@{}lcccccc@{}}
\toprule
Analysis & $n$ & $w_1$ & $\rho_1$ & $\rho_2$ & Entropy & Confident (\%)\\
\midrule
Primary (all records) & 8772 & 0.52 & 0.27 & 0.63 & 0.96 & 3.4\\
Excluding flagged records & 8648 & 0.49 & 0.25 & 0.62 & 0.96 & 2.4\\
Random ties (1) & 8772 & 0.43 & 0.22 & 0.60 & 0.95 & 1.9\\
Random ties (2) & 8772 & 0.52 & 0.26 & 0.63 & 0.96 & 3.5\\
\bottomrule
\end{tabular}
\end{adjustbox}
\end{table}

\subsection{Diagnostics}
Tables~\ref{tab:s-diag} and~\ref{tab:s-nhanes-diag} list rank-normalized split $\widehat R$ values \citep{Vehtari2021} and effective sample sizes. For the four-family fits by year of study, the largest $\widehat R$ of the two tail coefficients is 1.006 and the smallest effective sample size 627. The Gaussian mixtures with $K\ge2$ in the cross-validation folds of the runner data (Section~\ref{sec:s-running}) mixed less well, in some folds with a smallest effective sample size below 100, because their components are weakly separated; their predictive densities, which average over the posterior, are less affected than their component parameters, and no conclusion rests on them.

\begin{table}[htbp]
\centering\small
\setlength{\tabcolsep}{3.5pt}
\setstretch{1.1}
\caption{Convergence diagnostics of the application fits: largest rank-normalized split $\widehat R$ and smallest bulk effective sample size over weights, component Kendall's $\tau$ and tail coefficients (four-family fits, four or two chains) or over all parameters (runner fits, four chains).}
\label{tab:s-diag}
\begin{adjustbox}{max width=\textwidth}
\begin{tabular}{@{}lcc@{}}
\toprule
Fit & max $\widehat R$ & min ESS\\
\midrule
Training sample & 1.006 & 790\\
All students & 1.003 & 1016\\
Random ties (1) & 1.004 & 414\\
Random ties (2) & 1.002 & 469\\
Women & 1.019 & 618\\
Men & 1.006 & 339\\
Runners, Gaussian $K=1$ & 1.000 & 5061\\
Runners, Gaussian $K=2$ & 1.005 & 539\\
Runners, Student $t_5$ & 1.001 & 5233\\
\bottomrule
\end{tabular}
\end{adjustbox}
\end{table}

\begin{table}[t]
\centering\small
\setlength{\tabcolsep}{3.5pt}
\caption{Adults (NHANES): convergence diagnostics over all monitored quantities (weights, component Kendall's $\tau$, tail coefficients) for each fit; four chains for the full-sample and training fits of the four-family mixture and of the Gaussian copulas fitted to the training data, two chains otherwise.}
\label{tab:s-nhanes-diag}
\begin{adjustbox}{max width=\textwidth}
\begin{tabular}{@{}lccc@{}}
\toprule
Fit & Quantities & Max $\hat R$ & Min ESS\\
\midrule
All, four-family & 11 & 1.003 & 1396\\
All, Gaussian $K=1$ & 1 & 1.002 & 1834\\
All, Gaussian $K=2$ & 4 & 1.001 & 718\\
Training, four-family & 11 & 1.005 & 1007\\
Random ties (1) & 11 & 1.002 & 916\\
Random ties (2) & 11 & 1.002 & 748\\
Age 20--39, four-family & 11 & 1.004 & 463\\
Age 40--59, four-family & 11 & 1.002 & 1116\\
Age 60--80, four-family & 11 & 1.002 & 498\\
\bottomrule
\end{tabular}
\end{adjustbox}
\end{table}

\section{Wearable-device data and running performance}\label{sec:s-running}

This section applies the methods to a smaller, four-dimensional dataset. It contains Fitbit records and a supervised one-mile field test for first-year students at Oral Roberts University \citep{BroaddusData2018,Broaddus2021}. Of 581 students, 84 walked rather than ran the mile and were excluded, leaving 497 runners (281 women, 216 men). We use the semester averages of daily steps and of daily minutes in the device's fat-burn and cardio heart-rate zones, which in 2017 corresponded to 50--69\% and 70--84\% of an age-predicted maximum heart rate, and the mile time, reversed so that larger values mean faster running. The peak zone was not used, because it is zero or close to zero for most students and nearly collinear with the cardio zone (Spearman correlation 0.78 within sex). All variables are complete, and margins are standardized within sex.

Five-fold cross-validation (Table~\ref{tab:running}) favors models with heavier joint tails than a single Gaussian copula: a Student $t_5$ copula and a three-component Gaussian mixture improve the mean held-out log density by 0.033 and 0.039, about two standard errors. We use the $t_5$ copula, which is nearly as good and whose sampler mixes well, and check every conclusion against the Gaussian fits. The runner fits use two chains of 60000 iterations for each fold and four chains of 80000 iterations for the full sample, with a quarter of each chain discarded and thinning by ten.

\begin{table}[t]
\centering\small
\setlength{\tabcolsep}{3.5pt}
\caption{Runners. Top: five-fold cross-validated mean log copula density, with the difference from a single Gaussian copula and its paired standard error. Bottom: association of mile performance with each activity measure under the Student $t_5$ copula, as Kendall's $\tau$ and as partial correlation given the other two measures; posterior means and 95\% credible intervals.}
\label{tab:running}
\begin{adjustbox}{max width=\textwidth}
\begin{tabular}{@{}lcc@{}}
\toprule
Model & Mean log density & Difference (SE)\\
\midrule
Gaussian, $K=1$ & 0.338 & ---\\
Gaussian, $K=2$ & 0.357 & 0.019 (0.010)\\
Gaussian, $K=3$ & 0.377 & 0.039 (0.015)\\
Student $t_5$, $K=1$ & 0.371 & 0.033 (0.017)\\
\bottomrule
\end{tabular}
\end{adjustbox}

\vspace{0.8em}
\begin{adjustbox}{max width=\textwidth}
\begin{tabular}{@{}lcc@{}}
\toprule
Activity measure & Kendall's $\tau$ & Partial correlation\\
\midrule
Daily steps & 0.19 (0.13, 0.24) & 0.36 (0.27, 0.43)\\
Cardio-zone minutes & 0.08 (0.02, 0.14) & 0.20 (0.11, 0.29)\\
Fat-burn-zone minutes & $-$0.13 ($-$0.18, $-$0.07) & $-$0.38 ($-$0.45, $-$0.30)\\
\bottomrule
\end{tabular}
\end{adjustbox}
\end{table}

The three activity measures are positively correlated with one another, yet they relate to running in different directions. Holding the other two measures fixed, mile performance is positively associated with steps (partial correlation 0.36) and with cardio-zone minutes (0.20) but negatively with fat-burn-zone minutes ($-0.38$; Table~\ref{tab:running}). Figure~\ref{fig:s-running} shows the corresponding conditional probabilities of a mile time in the fastest quartile of the student's sex. With the other measures at their medians, moving from the lower to the upper quartile of steps raises this probability from 0.11 to 0.28, while the same move in fat-burn minutes lowers it from 0.30 to 0.10. A student in the upper quartile of steps and the lower quartile of fat-burn minutes has a 0.43 probability (95\% interval 0.39 to 0.47) of being among the fastest quarter, against 0.17 (0.15 to 0.20) for a student with the same steps and cardio minutes but upper-quartile fat-burn minutes.

\begin{figure}[t]
\centering
\includegraphics[width=\textwidth]{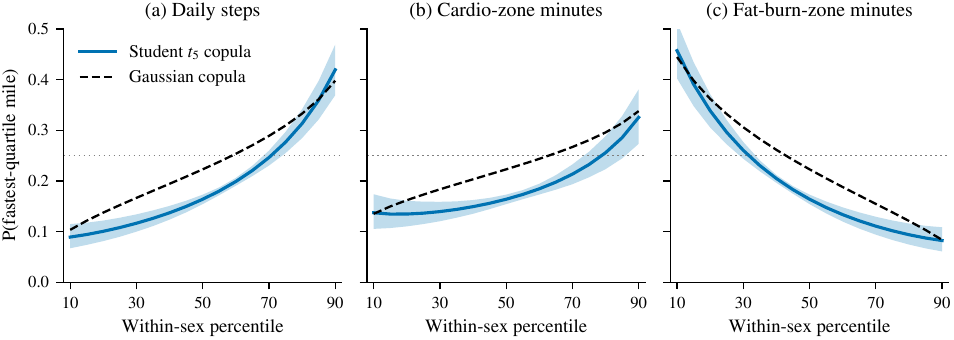}
\caption{Probability of a mile time in the fastest within-sex quartile among 497 runners, as one activity measure varies and the other two are held at their within-sex medians. Posterior predictive probabilities under a Student $t_5$ copula (solid, with 95\% band) and a Gaussian copula (dashed); the dotted line marks the unconditional value 0.25.}
\label{fig:s-running}
\end{figure}

The negative association has a physiological reading. Heart-rate zones are defined relative to an age-predicted maximum, and heart rate during a given activity falls as aerobic fitness rises \citep{Achten2003}. At a fixed step count, therefore, more time in the moderate zone indicates that everyday movement raises the student's heart rate more, which is a sign of lower fitness rather than of extra training. Vigorous cardio-zone minutes, which mostly reflect deliberate exercise, point the other way. For monitoring programs that build goals from wearable ``active minutes'', heart-rate-based minutes measure physiological strain as well as activity and should be read jointly with steps rather than added to them. These are cross-sectional associations among 497 students, and consumer devices measure heart rate and activity with error \citep{Feehan2018}; device placement, resting heart rate and unrecorded activity may also contribute.

\end{document}